\documentclass[final,12pt]{clear2026} 

\usepackage[capitalize,noabbrev]{cleveref}
\usepackage{mathtools} 

\usepackage[shortcuts]{extdash}

\usepackage{amsfonts}
\usepackage{thmtools} 
\usepackage{thm-restate} 

\usepackage{microtype}
\usepackage{booktabs} 
\usepackage{adjustbox}
\usepackage{paralist} 
\usepackage{nameref} 
\usepackage{graphicx}
\usepackage{caption}
\usepackage{wrapfig}

\crefname{assumption}{assumption}{assumptions}
\Crefname{assumption}{Assumption}{Assumptions}

\crefname{proposition}{proposition}{propositions}
\Crefname{proposition}{Proposition}{Propositions}

\newcommand{\indep}{\perp \!\!\! \perp} 
\newcommand{\mathdash}{\relbar\mkern-9mu\relbar}
\newcommand{\EV}[1]{\mathbb{E}\left[#1\right]} 

\newcommand{\bX}{\mathbf{X}}
\newcommand{\bU}{\mathbf{U}}
\newcommand{\bx}{\mathbf{x}}
\newcommand{\bF}{\mathbf{F}}
\newcommand{\bG}{\mathbf{G}}

\newcommand{\cG}{\mathcal{G}}
\newcommand{\XXX}{\mathcal{X}}
\newcommand{\YYY}{\mathcal{Y}}

\newcommand{\cdf}{\mathrm{P}}
\newcommand{\pdf}{\mathrm{p}}
\newcommand{\cdfQ}{\mathrm{Q}}

\newcommand{\Corr}{\mathrm{Corr}}

\renewcommand{\set}[1]{\{#1\}} 
\newcommand*\absoluteValue[1]{\left\lvert #1\right\rvert}

\DeclareMathOperator*{\argmin}{arg\,min}
\newcommand{\icmaxent}{i\=/CMAXENT\ } 

\newcommand{\ddo}{\mathrm{do}}

\newcommand\extrafootertext[1]{%
    \bgroup
    \renewcommand*{\footnoteseptext}{}
    \renewcommand\thefootnote{\fnsymbol{footnote}}%
    \renewcommand\thempfootnote{\fnsymbol{mpfootnote}}%
    \footnotetext[0]{#1}%
    \egroup
}

\usepackage{xcolor}

\usepackage{tikz}
\usetikzlibrary{shapes.geometric,positioning,bayesnet,calc}
\usetikzlibrary{arrows,shapes,plotmarks,positioning}
\usetikzlibrary{decorations.markings}
\tikzset{>=stealth'} 
\tikzstyle{graphnode} = 
   [circle,draw=black,minimum size=22pt,text centered,text
     width=22pt,inner sep=0pt] 
\tikzstyle{var}   =[graphnode,fill=white]
\tikzstyle{vardashed}   =[graphnode,draw=gray,fill=white]
\tikzstyle{obs}   =[graphnode,fill=black,text=white]
\tikzstyle{obsgrey}   =[graphnode,draw=white,fill=lightgray,text=black]
\tikzstyle{par}    =[graphnode,draw=white,fill=red,text=black] 
 \tikzstyle{crucial} =[graphnode,draw=white,fill=yellow,text=black] 
\tikzstyle{fac}   =[rectangle,draw=black,fill=black!25,minimum size=5pt]
\tikzstyle{facprior} =[rectangle,draw=black,fill=black,text=white,minimum size=5pt]
\tikzstyle{edge}  =[draw=white,double=black,very thick,-]
\tikzstyle{blueedge}  =[draw=white,double=blue,very thick,-]
\tikzstyle{rededge}  =[draw=white,double=red,very thick,-]
\tikzstyle{prior} =[rectangle, draw=black, fill=black, minimum size=
5pt, inner sep=0pt]
\tikzstyle{dirprior} = [circle, draw=black, fill=black, minimum
size=5pt, inner sep=0pt]
\tikzstyle{dot_node}=[draw=black,fill=black,shape=circle]

\title[Estimating Joint Interventional Distributions]{Estimating Joint Interventional Distributions from Marginal Interventional Data}

\usepackage{times}

\clearauthor{%
 \Name{Sergio Hernan {Garrido Mejia}} \Email{shgm@tuebingen.mpg.de}\\
 \addr Max Planck Institute for Intelligent Systems\\
 T\"ubingen, Germany
 \AND
 \Name{Elke Kirschbaum*}\\
 \addr Amazon\\
 T\"ubingen, Germany
 \AND
 \Name{Armin Keki\'c*}\\
 \addr Max Planck Institute for Intelligent Systems\\
 T\"ubingen, Germany
 \AND
 \Name{Bernhard Schölkopf} \\
 \addr Max Planck Institute for Intelligent Systems\\
 ELLIS Institute\\
 T\"ubingen, Germany
 \AND
 \Name{Atalanti Mastakouri}\\
 \addr Amazon\\
 T\"ubingen, Germany
}

\begin{document}

\maketitle
\extrafootertext{* denotes equal contribution.}

\begin{abstract}
In this paper we show how to exploit \emph{interventional} data to acquire the joint conditional distribution of all the variables using the Maximum Entropy principle.
To this end, we extend the Causal Maximum Entropy method to make use of data arising from identifiable interventional distributions in addition to data from the observational distribution.
Using Lagrange duality, we prove that the solution to the Causal Maximum Entropy problem with interventional constraints lies in the exponential family, as in the Maximum Entropy solution.
Our method allows us to perform two tasks of interest when marginal interventional distributions are provided for any subset of the variables.
First, we show how to perform parental discovery from a mixture of observational and single-variable interventional data, and, second, how to infer joint interventional distributions.
For the former task, we show on synthetically generated data, that our proposed method outperforms the state-of-the-art method on merging datasets, and yields comparable results to the KCI-test which requires access to joint observations of \emph{all} variables.
\end{abstract}


\section{Introduction}

Randomised Controlled Trials (RCTs) are the standard method for identifying the causal effect of a treatment on a target variable. 
However, their utility is often constrained by cost when studying complex interventions. 
Determining how a new drug, for instance, interacts with existing medications or medical procedures requires testing a number of combinations that grows exponentially. 
Furthermore, research goals frequently extend beyond the average treatment effect to understanding its impact under specific combinations of conditions (\textit{e.g.}, medical pre-conditions or risk factors). 
While this conditional knowledge is crucial and must be recorded during the trial, designing an RCT capable of considering all potentially relevant factors proves challenging. 
This difficulty often leads to research where multiple independent studies investigate different aspects or subsets of treatments and conditions. 
Consequently, researchers rarely have access to the full conditional distribution of the target variable or the joint interventional distribution of multiple treatments. 
This fundamental lack of joint data complicates the identification of whether a treatment or condition has a direct causal effect or merely an indirect influence through another factor.

As another real-world example beyond the medical domain, consider the problem of finding the effect of different fertilisers and planting methods on a particular crop yield \citep{hindersah2022rice}.
As the number of fertilisers and planting methods increases, the experimental design becomes prohibitively expensive due to combinatorial explosion.
Nevertheless, a researcher interested in the combined effect might have observational data and data from single experiments, such as nitrogen \citep{qiu2022effect} or potassium \citep{wihardjaka2022effect} fertilisers on crop yield.
To address this challenge, we propose a method for combining experimental and observational data to infer joint interventional distributions.

To do this, we extend the Causal Maximum Entropy (CMAXENT) principle~\citep{janzing2021causal} to the interventional-CMAXENT (\emph{i-CMAXENT}). 
We show that the resulting distribution for \icmaxent lies in the exponential family, similar to the traditional Maximum Entropy (MAXENT) distribution~\citep{wainwright2008graphical}.
This allows us to perform two tasks of interest:
(i) We can estimate joint interventional distributions from single variable interventions \citep{saengkyongam2020learning,kekic2025learning} in the marginal causal problem setting \citep{gresele2022causal}.
(ii) We can perform parental discovery \citep{peters2016causal,heinze2018invariant} under the causal marginal setting.
That is, under certain graph constraints, it allows us to infer the true causal parents of a variable of interest from a set of potential causes, even when the variables are not jointly observed.

The contributions of this paper are as follows:
\begin{itemize}
    \item We extend the Causal Maximum Entropy principle (CMAXENT) to use data from experimental conditions (\Cref{sec:icmaxent}).
    \item We prove that given a set of statistical constraints, the solution to the proposed method lies in an exponential family of distributions, extending the scope of these to include causal semantics (\Cref{sec:icmaxent}).
    \item We provide conditions under which this method can be used to combine observational and experimental data (\Cref{sec:icmaxent}), which is only possible nonparametrically under strict constraints (\Cref{sec:motivation}), and test the  results empirically against usual CMAXENT (\Cref{sec:experiments,sec:results}).
\end{itemize}

\section{Related work}
\label{application_related_work}

\paragraph{Marginal problem and causality} Various methods address the statistical problem of merging information from datasets with overlapping subsets of random variables, called the \emph{Marginal Problem} \citep{deming1940least,kellerer1964}.
The problem of combining information from overlapping data in causal structure learning has been studied in \citep{danks2008integrating,tillman2011learning, dhir2020integrating}.
However, the structures they are able to learn are based on conditional independence tests within the overlapping datasets; in other words, they have to observe variables jointly.
Recently, an extension of this problem was introduced, the \emph{Causal Marginal Problem}, where information from disjoint datasets is merged into a single \emph{causal} model \citep{gresele2022causal, sani2023bounding}.
\citet{gresele2022causal}, for example, address whether data from subsets of variables together with a known graph structure can be used to determine a set of joint SCMs that are counterfactually consistent with the marginal data.
Their work differs from ours in that we are interested in finding joint interventional distributions, while they focus on bounding counterfactual quantities, which allows them to falsify causal models. 
Further, they assume the causal graph to be given, while our approach can be used for parental discovery.

\paragraph{Joint interventional distributions} A recent object of research interest is the conditions under which joint causal effects can be identified from interventions on variable subsets, single variable interventions being the limiting case.
This question has been tackled from both a parametric and a nonparametric point of view.
From the nonparametric perspective, causal effect identification from experimental data was studied by \citet{bareinboim2012causal}, where a sound and complete algorithm was introduced for z-identification of causal effects under experimental conditions.
This was extended later by \citet{lee2020general}, where the $g$-identification formula was introduced in scenarios where not all experimental data is available.
\citet{jung2023estimating} presented an estimator based on the $g$-identification formula with properties against bias.
\citet{tikka2021causal} developed a sound search-based algorithm, where heuristics and search reduction techniques are able to decide (and provide a nonparametric estimator, if possible) for joint effect identification from observational and experimental data.

From the parametric perspective, \citet{saengkyongam2020learning} studies identification in nonlinear Gaussian models with confounding between the covariates and the outcomes, while \citet{kekic2025learning} provides an unbiased estimator of joint interventional effects from single variable interventions without distributional assumptions.
However, they rely on an additive outcome mechanism assumption.
\citet{gimenez2022causal}, explore the identifiability of joint causal effects under different assumptions of the problem such as linear and nonlinear models, the existence of instrumental variables and sparsity.
The potential outcome literature has also studied this question from a parametric point of view (see \citep{shi2023data,colnet2024causal} and references therein).

The complementary question of the conditions under which single variable interventions can be identified from joint causal effects has also been studied.
\citet{jeunen2022disentangling}, for example, find the conditions under which single variable interventions can be identified from joint interventions in confounded additive noise models.
In addition to the previous result, \citet{elahi2024identification} study how to obtain all possible causal effects from only some joint interventions in additive noise models with Gaussian noise.

\paragraph{Causal discovery} The task of parental discovery has been addressed from various perspectives, with the most prominent ones being those based on invariant causal prediction \citep{peters2016causal,heinze2018invariant}.
While these methods are powerful and do not assume knowledge about which variables in the system were intervened upon, they cannot operate in the setting of marginally observed sets of variables, such as the causal marginal problem.

Furthermore, other research has focused on causal discovery using observational and experimental data using either Bayesian inference methods \citep{cooper1999causal, eaton2007exact}, causal graph conditions \citep{tian2001causal}, meta-analysis methods for combining $p$-values \citep{tillman2009structure}, or constraint-based methods \citep{triantafillou2015constraint}.
A particularly important idea in this line of research is the Joint Causal Inference (JCI) framework \citep{mooij2020joint}, where these ideas were unified.
However, none of these methods work under the causal marginal problem, which limits their application in heavy missing data problems.

A method that allows for parental discovery, even under this marginal setting is CMAXENT \citep{garrido2022obtaining}.
However, CMAXENT can only make use of observational data, limiting the amount of information that it can leverage to find the relevant causal parents.
An extensive comparison of our approach with CMAXENT can be found in \cref{sec:results}.

\section{Motivation}
\label{sec:motivation}

In this section we will prove that even in cases where the parents of the target variable are confounded we can identify nonparametric joint interventional effects from observational or interventional marginal data. 
This holds true even in cases where these joint interventional distributions cannot be identified using the rules of do-calculus \citep{pearl2009causality}.

\begin{figure}
    \centering

\begin{tikzpicture}
\node[latent] (U) {$U$};
\node[below=1cm of U] (placeholder) {};
\node[obs,right=.5cm of placeholder] (X2) {$X_2$};
\node[obs,left=.5cm of placeholder] (X1) {$X_1$};
\node[obs,below=1cm of placeholder] (Y) {$Y$};

\edge{U}{X1, X2};
\edge{X1}{Y};
\edge{X2}{Y};

\end{tikzpicture}
    \caption{Graph that allows nonparametric marginal data combination}
    \label{fig:nonparametric-basic-graph}
\end{figure}
Consider a causal system with binary variables $X_1,X_2,Y$ and $U$, where $U$ is an unobserved confounder, and the joint distribution $\pdf(\bX,Y) := \pdf(X_1,X_2,Y)$ is Markov relative to the graph in \Cref{fig:nonparametric-basic-graph}.
We are interested in $\pdf(Y\mid \ddo(X_1=x_1,X_2=x_2))$, which in the case of sufficient causal systems would correspond to $\pdf(Y\mid X_1=x_1,X_2=x_2)$.
We have $\pdf(Y\mid X_1)$, $\pdf(Y\mid X_2)$ and $\pdf(\bX)$ as data.
Using a modern interventional distribution identification engine \citep{tikka2021causal}, we verify that the joint interventional distribution is not identifiable using the given data.
Nonetheless, we can still identify the joint interventional distribution nonparametrically:
\begin{restatable}[Nonparametric identification through observational data combination]{proposition}{nonparametricIdentification}\label[proposition]{prop:nonparametricIdentification}
Let $\bX, Y$ be binary random variables and suppose $\pdf(\bx)>0$, for all $\bx$.
Then the joint interventional distribution $\pdf(Y\mid \ddo(X_1=x_1,X_2=x_2))$ is identifiable using $\pdf(\bX)$ and $\pdf(Y\mid X_i)$ for $i=1,2$.
\end{restatable}
\begin{remark}
    In the statement above, it is not necessary that $X_i$ does not cause $X_j$. 
\end{remark}

The proof is based on the fact that there are four equations, namely one for each $\pdf(Y \mid X_i = x_i)$, and four unknowns $\pdf(Y \mid X_1=x_1, X_2=x_2)$.
All proofs can be found in \Cref{app:proofs}.
In the particular case of the graph in $\Cref{fig:nonparametric-basic-graph}$ the backdoor criterion holds, so that if we had interventional distributions instead of observational ones, the target distribution is also identifiable, giving us the following Corollary:
\begin{corollary}[Nonparametric identification through single variable intervention data combination]
Let $\bX, Y$ be binary random variables and suppose $\pdf(\bx)>0$, for all $\bx$.
Then the joint interventional distribution $\pdf(Y\mid \ddo(X_1=x_1,X_2=x_2))$ is identifiable using $\pdf(\bX)$ and the single-variable interventions given by $\pdf(Y\mid \ddo(X_i=x_i))$ for $i=1,2$, as long as each $\pdf(Y\mid \ddo(X_i=x_i))$ is identifiable.
\end{corollary}

We see how we cannot apply this idea for more than two variables using either only observational or interventional data.
Indeed, the number of unknowns grows exponentially with the number of variables, whereas the number of equations grows linearly. 
Nevertheless, the same result can be applied once more by combining observational and interventional data in the case of four variables:
\begin{corollary}
Let $\bX = \set{X_i: i=1,2,3,4}, Y$ be binary random variables and $\pdf(\bx)>0$ for all $\bx$.
If $\pdf(Y\mid \ddo(X_i=x_i))$ are identifiable using the backdoor criterion for all $i$, we can identify the joint interventional distribution $\pdf(Y \mid \ddo(\bX))$ using $\pdf(\bX)$, $\pdf(Y\mid X_i)$ and $\pdf(Y\mid\ddo(X_i=x_i))$ for all $i$.
\end{corollary}
This result motivates our search for methods that allow us to combine both observational and interventional data beyond four variables.

\section{Causal Maximum Entropy}
\label{sec:cmaxent}

In this section, we formally state the Maximum Conditional Entropy problem (\citet[Chapters 8 and 20]{koller2009probabilistic}, \citet{berger1996maximum}), and its causal interpretation, Causal Maximum Entropy (CMAXENT) \citep{sun2006causal,janzing2021causal,garrido2022obtaining}, which we will use as a basis for our proposed model in \Cref{sec:icmaxent}.

Consider a set of $D$ random variables $\bX=\{X_{1},\dots,X_{D}\}$ of \emph{potential causes}, and an \emph{effect} $Y$, with realisations $\bx\in\mathcal{X}$ and $y\in\mathcal{Y}$.
We define $\cdf$ to be a probability measure on $\YYY\times \XXX$.

Let $\bF=\{f_{k}\}$ be a set of $K$ \emph{marginal} measurable functions with $f_{k}:\mathcal{Y}\times\mathcal{X}_{S_{k}} \rightarrow \mathbb{R}$, where $S_{i}\subseteq\{1, \dots,D\}$ is an index set of potential causes. 
For these functions, we are given empirical averages denoted by
$\tilde{f}_{k} := \frac{1}{N}\sum^{N}_{i=1}f_{k}(y^{i},\bx^{i}_{S_{k}})$, where $N$ is the number of observations of $y^i$ and $\bx_{S_k}^i$.
The expectations of the marginal functions can be computed with respect to any valid joint density of $Y$ and $\bX$: $\mathbb{E}[f_{k}(Y,\bX_{S_k})]=\sum_{y,\bx}\pdf(y,\bx)f_{k}(y,\bx_{S_{k}})$. 

Furthermore, let $J$ be a set of \emph{conditional} measurable functions $\bG=\{g_{j}\}$ with $g_{j}:\mathcal{Y}\times\mathcal{X}_{S_{j}} \rightarrow \mathbb{R}^{\lvert\bX_{S_j}\rvert}$.
For conditional functions, the empirical averages $\tilde{g}_{j}(\bx_{S_{j}}) := \frac{1}{N}\sum^{N}_{i=1}g_{j}(y^{i},\bx_{S_{j}})$ depend on the value of the conditioning variables $X_{S_{j}}$, where the conditioning set, $S_{j}$ is different for every conditional function $g_{j}$.
The expectations are computed with respect to a valid conditional distribution: ${\mathbb{E}[g_j(Y,\bx_{S_j}) \mid \bX_{S_j} = \bx_{S_j}]= \sum_{y}\pdf(y\mid\bx_{S_{j}})g_{j}(y,\bx_{S_{j}})}$.

Suppose we are given a density of the potential causes $\pdf(\bX)$ and the sets of functions $\bF$ and $\bG$ with their respective empirical averages, but we do not know the conditional density $\pdf(Y\mid \bX)$. 
Then the principle of Maximum Conditional Entropy suggests choosing the conditional density $\pdf_{\lambda}(Y\mid \bX)$ that has expectations consistent with the given empirical averages and also maximises the Shannon conditional entropy $H(Y\mid X)$ \citep{jaynes1957information,berger1996maximum,farnia2016minimax}:
\begin{align}\label{eq:cmaxentopt}
    \begin{split}
    \max_{\pdf(y\mid\bx)}&\quad H(Y\mid\bX) = -\sum_{y,\bx}\pdf(y\mid\bx)\pdf(\bx)\log \pdf(y\mid\bx)\\
    \text{s.t.}\quad & \mathbb{E}[f_{k}(Y,\bX_{S_k})]=\tilde{f}_{k},\,\text{ for all } k=1,\dots,K\\
    & \mathbb{E}[g_{j}(Y,\bx_{S_j}) \mid \bX_{S_j} = \bx_{S_j}]=\tilde{g}_{j}(\bx_{S_{j}}),\,\text{ for all } j=1,\dots,J\\
    & \sum_{y}\pdf(y\mid\bx)=1,\,\text{ for all } \bx.
    \end{split}
\end{align}
Solving this optimisation problem, using the Lagrange multiplier formalism, yields an exponential family distribution:
\begin{align}\label{eq:condmaxentsol}
    \pdf_{\lambda}(y\mid\bx) =
    \exp \Biggl( \sum_{k=1}^{K}\lambda_{k}f_{k}(y,\bx_{S_k})
    + \sum_{j=1}^{J}\sum_{\bx_{S_{j}}}\lambda^{\bx_{S_{j}}}_{j} g_{j}(y,\bx) + \alpha(\bx)\Biggr),
\end{align}
where $\alpha(\bx)$ is the normalising constant and $\lambda=\{\lambda_{k},\lambda^{\bx_{S_{j}}}_{j}\}$ are the Lagrange multipliers.

Intuitively, this density with maximum entropy is as close as possible to the uniform while satisfying the expectation constraints.
\citet[Chapter 11]{jaynes2003probability} interprets MAXENT as a way to find a density without introducing more information than given by the data.
Critically, MAXENT, Maximum Conditional Entropy, or CMAXENT do not aim at estimating the ``true distribution'' of the data, but instead provide a guess given certain statistical properties of the data.
The resulting MAXENT distribution posses powerful statistical properties that have been deeply studied in the literature \citep{grunwald2004game, wainwright2008graphical, farnia2016minimax}.

In CMAXENT \citep{sun2006causal, janzing2009distinguishing, janzing2021causal, garrido2022obtaining}, causal semantics are introduced via graphical models.
In other words, we assume that the variables in our system have some cause-effect relation that can be represented by a causal graph \citep{pearl2009causality}.
In CMAXENT, the densities are computed in the causal order given by a (hypothesised) causal graph.
For example, if $\bX$ are potential causes of $Y$, we first find the MAXENT density of $\bX$ subject to constraints on $\bX$, and then the density with Maximum Conditional Entropy of $Y$ given $\bX$ using the found $\pdf(\bX)$ and subject to the constraints associated with $\bX$ and $Y$.
Because of the introduced causal semantics, we can use expectations involving interventions in the estimation of the density $\pdf_{\lambda}(y\mid\bx)$, which corresponds to the so-called Independent Causal Mechanisms \citep{scholkopf2012causal} of $Y$ given its potential causes $\bX$.

The use of interventional data would not be possible for non-causal Maximum Conditional Entropy, since the conditional distribution there is devoid of any causal structure and hence cannot be used for operations of the causal hierarchy like interventions or counterfactuals \citep{pearl2018theBook}. 
In \cref{sec:icmaxent}, we exploit this observation to extend the CMAXENT method to use data from interventions.

\section{Interventional CMAXENT (\icmaxent)}\label{sec:icmaxent}

In this section, we introduce \icmaxent, the modification of CMAXENT to include data on interventional distributions, and prove that the solution to the optimisation problem is an exponential family distribution, as in traditional MAXENT. 

Before formally introducing \icmaxent, we present the necessary assumptions needed to use interventional data in the optimisation problem.

\begin{assumption}[Weak causal sufficiency]
\label[assumption]{asmp:weakerSufficiency}
    We assume there are no unobserved confounders between the effect $Y$ and its potential causes $\bX$.
    However, there \emph{can} exist hidden confounders among the potential causes.
\end{assumption}
\Cref{asmp:weakerSufficiency} is needed to exclude any backdoor paths between any potential cause and the target $Y$.
More specifically, if hidden confounders exist between $\bX$ and $Y$, then we cannot phrase the interventional moments in terms of the distributions of observed nodes alone.

\begin{assumption}[Positivity of the potential causes]
\label[assumption]{asmp:probabilityZero}
    We assume $\pdf(\bx)>0$ for all $\bx$.
\end{assumption}
\Cref{asmp:probabilityZero} is required to ensure well-defined conditional distributions for all values $\bx$ of $\bX$.

We also assume Faithful $f$-expectations as in \citet[Definition 1]{garrido2022obtaining}, this assumption is used to relate the results of statistical estimation with a causal graph, as in the usual faithfulness assumption.
\begin{assumption}[Faithful $f$-Expectations]
\label[assumption]{asmp:faithfulFexpectations}
A distribution $\cdf$ is said to have faithful $f$-expectations relative to a DAG $\cG$, if for any distribution $\cdfQ$ Markov relative to $\cG$ it holds that whenever $\lambda_{f_k}^\cdfQ\neq 0$, then $\lambda_{f_k}^\cdf\neq 0$.
\end{assumption}

For an intuitive explanation, consider a triplet $X, Y, Z$ drawn from a distribution $\cdf$ Markov relative to a DAG $\cG$, for which $X \not\indep_\cdf Y \mid Z$.
Then according to \Cref{asmp:faithfulFexpectations}, the dependence of $X$ and $Y$ after observing $Z$, will also hold in the projected space of the MAXENT distribution. 
By contrapositive, an independence in the space of MAXENT will also hold in the original distribution space.
In \Cref{app:faithfulness} we discuss other definitions of faithfulness in the literature and how Faithful $f$-expectations relates to these different notions.

Next, let us introduce the required additional notation.
Let $\mathbf{H}=\set{h_l}$ be a set of $L$ \emph{interventional functions} for which the empirical averages are denoted by $\tilde{\mathbf{H}}$.
Further, suppose that, in addition to a conditioning set $S^C_l$, we also have a set $S^I_l$ of variables that were intervened upon.
For these functions, the expectations are computed with respect to the interventional density $\pdf(Y\mid \ddo(\bX_{S^I_l}), \bX_{S^C_l})$:
\begin{align}\label{eq:int_function_expectation}
    \mathbb{E}&[h_l(Y, \bx_{S^I_l}, \bx_{S^C_l}) \mid \ddo(\bX_{S^I_l} = \bx_{S^I_l}), \bX_{S^C_l} = \bx_{S^C_l}]
    = \sum_{y} \pdf(y\mid \ddo(\bx_{S^I_l}), \bx_{S^C_l}) h_l(y,\bx_{S^I_l}, \bx_{S^C_l}).
\end{align}
We assume $S^C_{l}\cap S^I_{l}=\emptyset$ for all indices $l$, as otherwise it would imply that the variables in the distributions we used to compute the empirical averages are both conditioned and intervened upon.

The reasoning behind the extension to \icmaxent is the following:
for a given interventional function $h_l$, if the corresponding interventional density ${\pdf(Y \mid \ddo(\bX_{S^I_l}), \bX_{S^C_l})}$ is identifiable, we can express the expectation in \Cref{eq:int_function_expectation} using the observational densities $\pdf(Y \mid \bX)$ and $\pdf(\bX)$.
In this case, we can estimate $\pdf_\lambda(Y \mid \bX)$ by fitting the parameters $\lambda$ such that the interventional expectations under $\pdf_\lambda(Y \mid \bX)$ are as close as possible to those given as constraints.
In other words, we leverage graphical nonparametric identification to write the interventional expression in terms of observational probabilities, allowing us to introduce interventional data as constraints for the joint density $\pdf_{\lambda}(Y, \bX)$.
As above, we find the distribution that maximises the conditional entropy such that its interventional expectations are close to the observed empirical averages.

We obtain the optimisation problem for \icmaxent by adding the following constraint to \Cref{eq:cmaxentopt}:
\begin{align}\label{eq:interventionalMAXENT}
    \mathbb{E}&[h_l(Y,\bx_{S^I_l},\bx_{S^C_l}) \mid \ddo(\bX_{S^I_l}=\bx_{S^I_l}), \bX_{S^C_l}=\bx_{S^C_l}]
    = \tilde{h}_l(\bx_{S^I_l}, \bx_{S^C_l}), \text{ for all } l=1, \dots, L.
\end{align}
Note that even though ${\pdf(y \mid \ddo(\bx_{S^I_l}), \bx_{S^C_l})}$ appears as part of the constraints, we consider it to be identifiable and hence can compute this interventional distribution as a functional of observational distributions.

By including empirical averages that come from interventions on $\bX$, we can use a richer class of data sources in comparison to CMAXENT, which could only use observational data.
In the following theorem, we study the resulting distribution of the \icmaxent optimisation problem.

\begin{restatable}[Exponential family of i-CMAXENT]{theorem}{interventionalMaxent}\label{th:interventionalMaxent}
Using the Lagrange multiplier formalism, the solution of \Cref{eq:cmaxentopt} with the additional constraint from \Cref{eq:interventionalMAXENT} is given by the following exponential family:
    \begin{align}
    \begin{split}
       \pdf_{\lambda}(y\mid \bx) 
       =& \exp \Biggl( \sum_{k=1}^{K}\lambda_{k}f_{k}(y,\bx)
        + \sum_{j=1}^{J}\sum_{\bx_{S_{j}}} \lambda^{\bx_{S_{j}}}_{j}g_{j}(y, \bx) \\
        &+ \sum_{l=1}^{L}\sum_{\bx_{S^I_l},\bx_{S^C_l}} \lambda^{\bx_{S^I_l}, \bx_{S^C_l}}_l h_l(y,\bx_{S^I_l}, \bx_{S^C_l}) + \beta(\bx) \Biggr), 
    \end{split}
    \end{align}
where $\beta(\bx)$ is the normalising constant, as in the conditional case.
\end{restatable}

In \Cref{th:interventionalMaxent}, we show that the well-known exponential family solutions of MAXENT and Maximum Conditional Entropy \citep{wainwright2008graphical, farnia2016minimax} can be generalised to the \icmaxent solution in an intuitive way.
The proof is shown in \Cref{app:proofs}.

\paragraph{\icmaxent for parental discovery}

In \citet{garrido2022obtaining}, we showed that causal edges can be inferred from the estimated Lagrange multipliers of the solution of the CMAXENT problem. 
Since we have \Cref{asmp:faithfulFexpectations}, which is required for those results to hold, and the solution of \icmaxent is an exponential family distribution, we can also use \icmaxent to infer causal edges.
Note that we need to know which interventional distributions are identifiable without requiring knowledge of which of the \emph{potential} causes \emph{actually} has a direct causal link to the effect variable if we want to use \icmaxent for parental discovery. 
The following proposition shows that this is possible.

\begin{restatable}[Identifiability and adjustment set of variables with only incoming arrows]{proposition}{identifiabilityIncomingArrows}
    \label[proposition]{prop:identifiabilityIncomingArrows}
    Let $\bX$ be a set of candidate causal parents of $Y$, which can be confounded.
    Assume we know the density $\pdf(\bX)$.
    If the only child of $X_j \in \bX$ is potentially, but not necessarily, $Y$, then $\pdf(Y \mid \ddo(X_j))$ is identifiable and a valid adjustment set for the atomic intervention is $\bX \setminus X_j$.
    That is, the rest of the potential causes.
\end{restatable}

A different way of thinking about \Cref{prop:identifiabilityIncomingArrows} is that we can decide whether we can use interventional data in an \icmaxent estimation only by looking at the arrows between $\bX$ and without any knowledge of the arrows between $\bX$ and $Y$.

\section{Experiments}
\label{sec:experiments}

\begin{figure*}[t]
    \subfigure[Structure graph (a)\label{fig:graph_exp_a}]{\adjustbox{max width=.33\textwidth}{

\begin{tikzpicture}
\node[latent](U){$U_1$};
\node[obs,below=1cm of U](X1){$X_1$};
\node[obs,right=.5cm of X1](X2){$X_2$};
\node[obs,right=.5cm of X2](X3){$X_3$};
\node[obs,right=.5cm of X3](X4){$X_4$};
\node[obs,right=.5cm of X4](X5){$X_5$};
\node[latent,above=1cm of X4](V){$U_2$};
\node[obs,below=1cm of X3](Y){$Y$};

\edge{U}{X1,X2};
\edge{V}{X3,X4};
\edge[dashed]{X1}{Y};
\edge[dashed]{X2}{Y};
\edge[dashed]{X3}{Y};
\edge[dashed]{X4}{Y};
\edge[dashed]{X5}{Y};

\end{tikzpicture}}}
    \subfigure[Structure graph (b)\label{fig:graph_exp_b}]{\adjustbox{max width=.33\textwidth}{

\begin{tikzpicture}
\node[latent](U1){$U_1$};
\node[latent,right=.5cm of U1](U2){$U_2$};
\node[latent,right=.5cm of U2](U3){$U_3$};
\node[latent,right=.5cm of U3](U4){$U_4$};
\node[latent,right=.5cm of U4](U5){$U_5$};
\node[obs,below=1cm of U1](X1){$X_1$};
\node[obs,below=1cm of U2](X2){$X_2$};
\node[obs,below=1cm of U3](X3){$X_3$};
\node[obs,below=1cm of U4](X4){$X_4$};
\node[obs,below=1cm of U5](X5){$X_5$};
\node[obs,below=1cm of X3](Y){$Y$};

\edge{U1}{X1, X2, X3};
\edge{U2}{X2, X3, X4};
\edge{U3}{X3, X4, X5};
\edge{U4}{X4, X5, X1};
\edge{U5}{X5, X1, X2};
\edge[dashed]{X1}{Y};
\edge[dashed]{X2}{Y};
\edge[dashed]{X3}{Y};
\edge[dashed]{X4}{Y};
\edge[dashed]{X5}{Y};

\end{tikzpicture}}}
    \subfigure[Structure graph (c)\label{fig:graph_exp_c}]{\adjustbox{max width=.33\textwidth}{

\begin{tikzpicture}
\node[latent] (U) {$U_1$};
\node[obs,below=1cm of U] (X3) {$X_3$};
\node[obs,right=.5cm of X3] (X5) {$X_5$};
\node[obs,left=.5cm of X3] (X1) {$X_1$};
\node[obs,left=.5cm of X1] (X2) {$X_2$};
\node[obs,right=.5cm of X5] (X4) {$X_4$};
\node[obs,below=1cm of X3] (Y) {$Y$};

\edge{U}{X1, X2, X3, X4, X5};
\edge{X1}{X2, X3};
\edge{X5}{X3, X4};
\edge[dashed]{X1}{Y};
\edge[dashed]{X2}{Y};
\edge[dashed]{X3}{Y};
\edge[dashed]{X4}{Y};
\edge[dashed]{X5}{Y};

\end{tikzpicture}}}

    \caption{
    (\cref{fig:graph_exp_a}), (\cref{fig:graph_exp_b}), and (\cref{fig:graph_exp_c}) show the graph structures used for our synthetic experiments.
    We randomise the presence of the edges in the lower part of the graphs (dashed arrows).
    The solid arrows are always present in the shown way.}
\end{figure*}

We test \icmaxent both for parental discovery and for joint interventional distribution estimation. 
In all experiments, we find the Lagrange multipliers of the exponential family distribution by minimising the norm of the residuals between empirical averages and the corresponding expectations (see  \Cref{app:normMinimisation} for details on the norm minimisation and \Cref{app:convergence} for details about the optimisation convergence).

\textbf{Synthetic data generation}
All used causal graphs comply with \Cref{asmp:weakerSufficiency} and consist of three levels:
unobserved confounders $\bU$, potential causes $\bX$, and an effect $Y$. 
While we assume causal sufficiency for the lower part of the graph, that is, no hidden confounders between $\bX$ and $Y$, we do allow for hidden confounders $\bU$ to exist among the potential causes $\bX$.
The considered graph structures are shown in \Cref{fig:graph_exp_a,fig:graph_exp_b,fig:graph_exp_c}.
Further details on synthetic data generation are in \Cref{app:detailsDGP}.

\paragraph{Causal discovery} 
The main question of this task is: 
from a set of \emph{potential} causes $\bX$ of our variable of interest $Y$, which of those variables are \emph{actual} causal parents?
To keep our model as comparable as possible with previous work on CMAXENT, we use the same graph structures as in \citet{garrido2022obtaining}.
For each structure in \Cref{fig:graph_exp_a,fig:graph_exp_b,fig:graph_exp_c}, we sample 200 random graphs as explained in \Cref{app:detailsDGP}.
When randomising the true causes of $Y$, we ensure that there is at least one causal parent and that at least one potential parent is not a true cause, so that the ROC curve is always defined.

\paragraph{Parental discovery benchmarking \icmaxent, CMAXENT, and KCI}
In this setting, we compare \icmaxent against CMAXENT and the Kernel Conditional Independence (KCI) test \citep{zhang2011kernel}.
For \icmaxent, we use constraints on the single-variable interventional densities $\pdf(Y\mid \ddo(X_i))$ (in the case of binary variables, this coincides with $\EV{Y \mid \ddo(X_i)}$) for all potential causes $X_{i}$. 
Since we have five potential causes in each graph, this means we have data on five interventional distributions. 
For CMAXENT, we use constraints on the marginal conditional densities $\pdf(Y\mid X_{i})$ (in the case of binary variables, this coincides with $\EV{Y \mid X_i}$) for each of the five potential causes $X_{i}$.
Hence, we again have data on five conditional distributions in this case. 
For both \icmaxent and CMAXENT, we consider two cases: 
in the first case, we are also given data on the full joint observational density $\pdf(\bX)$. 
In the second case, only the marginal observational densities $\pdf(X_1), \dots, \pdf(X_5)$ are given.

In the second case, we estimate $\pdf(\bX)$ by merging the constraints on the marginals, also using MAXENT.
We generate 100 observations per case and per sampled graph.
From these samples, we compute the empirical averages used as constraints in the optimisation procedure.
For the KCI test, we assume that KCI has access to 1,000 samples from the joint density $\pdf(\bX, Y)$.
This does not reflect the causal marginal problem we are considering in this paper. 
In fact, KCI could not be run in such a scenario.
\footnote{Similarly, ICP \citep{peters2016causal} is not designed for the causal marginal problem.
Hence, we do not compare against it.} 
Nevertheless, we benchmark our method against KCI to show what performance can be achieved on this task with access to the full joint observational distribution.

\paragraph{Parental discovery with combinations of interventional and conditional information using \icmaxent}
In this setting, we evaluate the performance of \icmaxent for parental discovery when interventional information is available for a fraction of the variables, while for the other variables we only have conditional information. 
With this experiment, we want to emulate a scenario that often occurs in real datasets:
namely, that not all variables are intervenable.
In this scenario, we assume that we are given data on the full joint observational distribution $\pdf(\bX)$ of the potential causes.

We use \Cref{prop:identifiabilityIncomingArrows} to decide which potential parent can provide interventional instead of conditional data.
For example, in \Cref{fig:graph_exp_c}, it follows from \Cref{prop:identifiabilityIncomingArrows} that we can use $\pdf(Y \mid \ddo(X_2))$ as a constraint regardless of the existence of the dashed edges because it is identifiable.
However, $\pdf(Y \mid \ddo(X_1))$ is not identifiable if the edge between $X_1$ and $Y$ exists.
As a result, we can only use conditional expectations for $X_1$.

In both settings, we use the relative difference estimator defined in \citet{garrido2022obtaining} as the parameter for the ROC curves of CMAXENT and \icmaxent (for details see \Cref{app:relativeDifference}). 
For the ROC curve of the KCI test, we vary the threshold on the $p$-value of the test.

\paragraph{Joint interventional distributions from single interventions} 
We are interested in finding multivariate interventional distributions when our available data comes from single-variable interventions.
For this task, we use the DAG shown in \Cref{fig:graph_exp_a} and sample 200 random graph instantiations.

We perform five experiments to assess how varying the type and amount of constraints influences the estimation of the joint interventional density $\pdf(Y\mid \ddo(X_1, X_2))$. 
In all of these cases, we use 1,000 observations to compute the empirical averages:
(i) Two potential causes $X_i$ and $X_j$ are chosen at random, and we provide \icmaxent with constraints on $\pdf(Y \mid X_i, X_j)$.
Additionally, we provide constraints on single-variable interventional distributions for the rest of the potential causes; that is, $\pdf(Y \mid \ddo(X_k))$ for all $k \neq i,j$.
(ii) We only provide constraints on $\pdf(Y \mid X_i, X_j)$.
(iii) We only provide constraints on the single-variable interventional densities for all potential causes $\pdf(Y \mid \ddo(X_i))$ for all $i$. 
(iv) We only provide constraints on the single-variable conditional densities $\pdf(Y \mid X_i)$ for all potential causes.
This scenario coincides with CMAXENT.
(v) As an additional baseline, we finally estimate MAXENT without constraints.
We then estimate the joint interventional density ${\pdf(Y\mid \ddo(X_1, X_2))}$ for each graph and plot the residuals against the true distribution, as shown in \Cref{fig:violin_true_vs_estimated}.

\section{Results}
\label{sec:results}

\begin{figure*}
    \subfigure[ROC curve for graph (a) \label{fig:roc_overlay_a}]{\includegraphics[width=.33\textwidth]{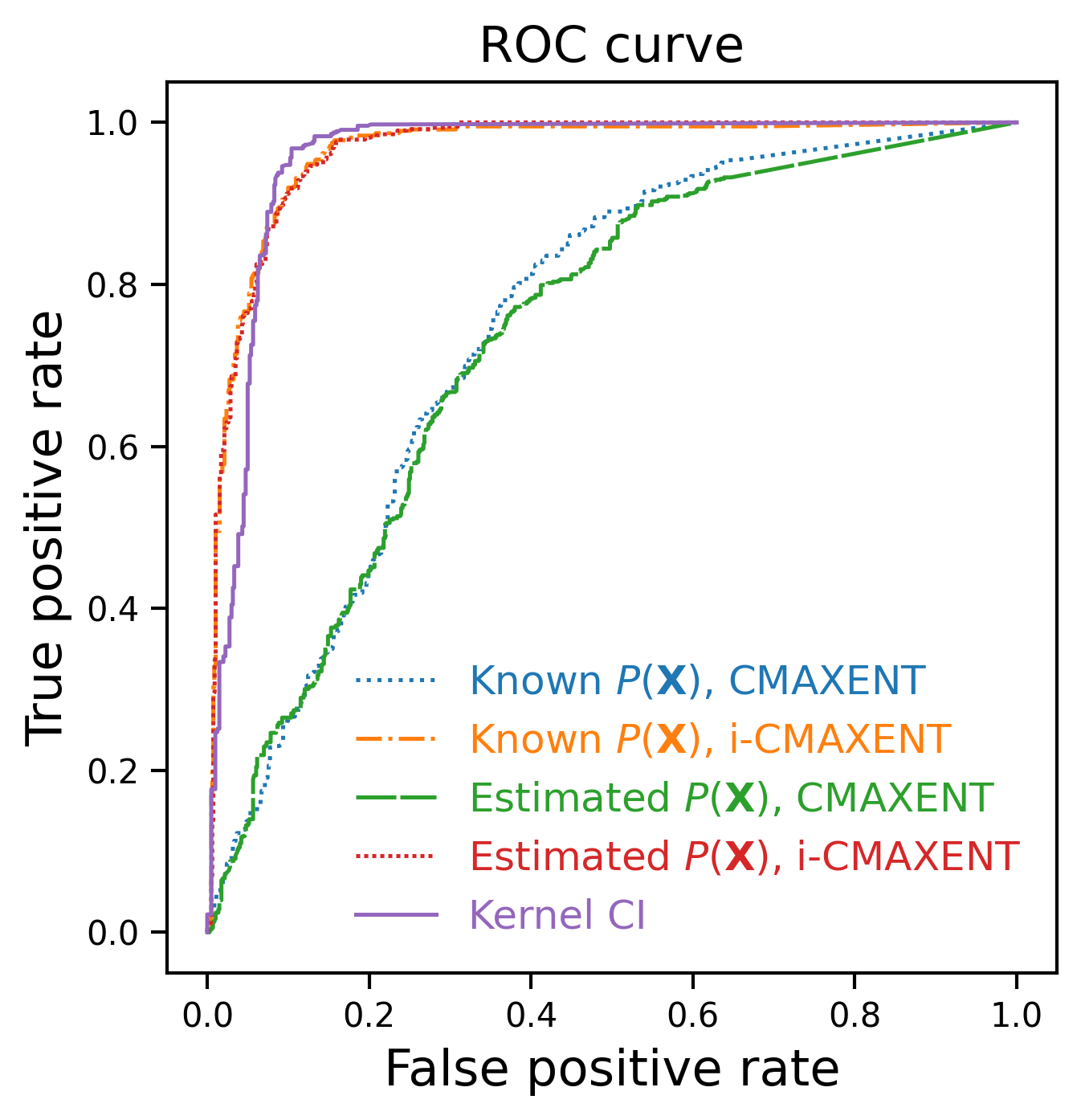}}
    \subfigure[ROC curve for graph (b) \label{fig:roc_overlay_b}]{\includegraphics[width=.33\textwidth]{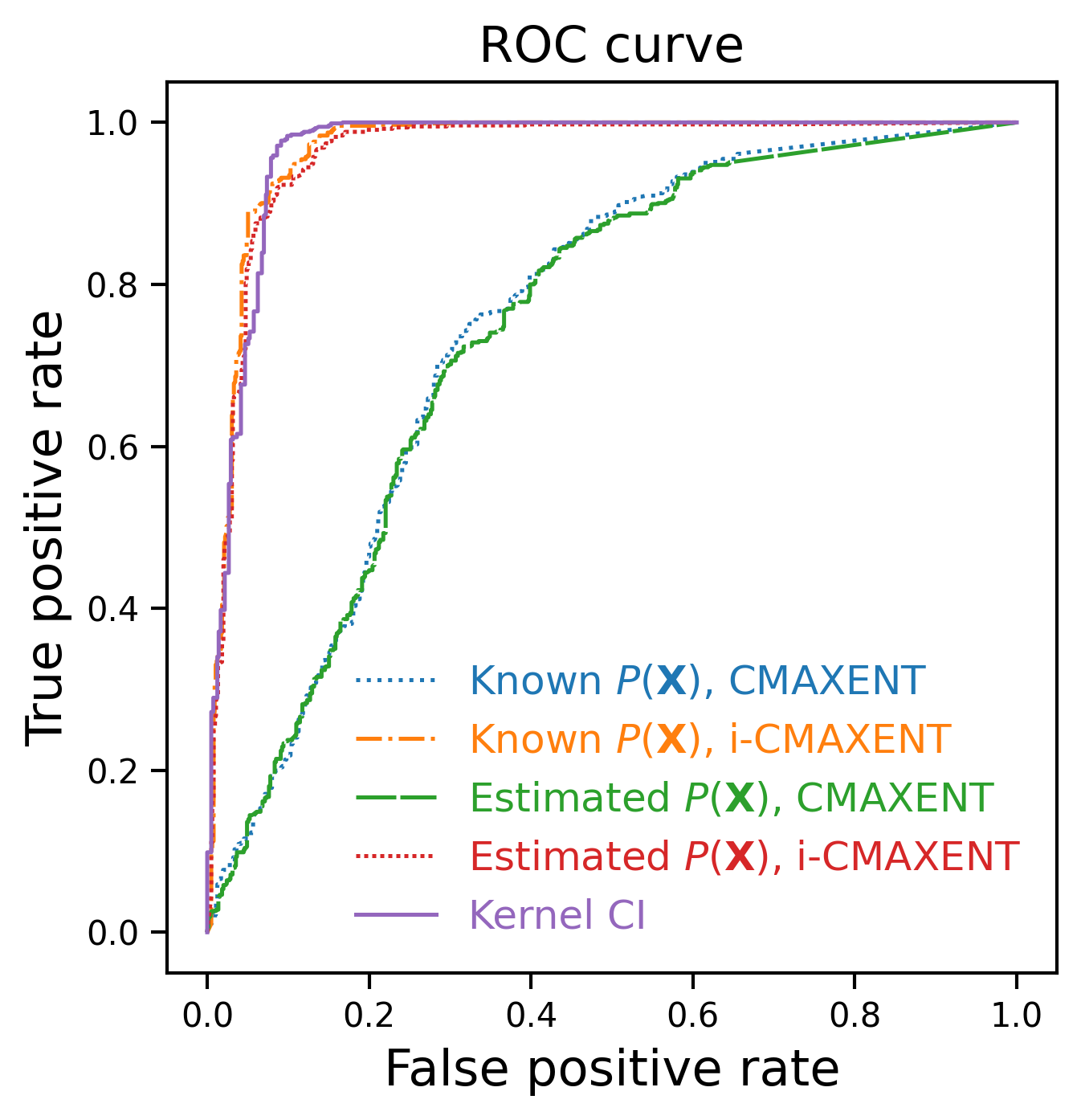}}
    \subfigure[ROC curve for graph (c) \label{fig:roc_overlay_c}]{\includegraphics[width=.33\textwidth]{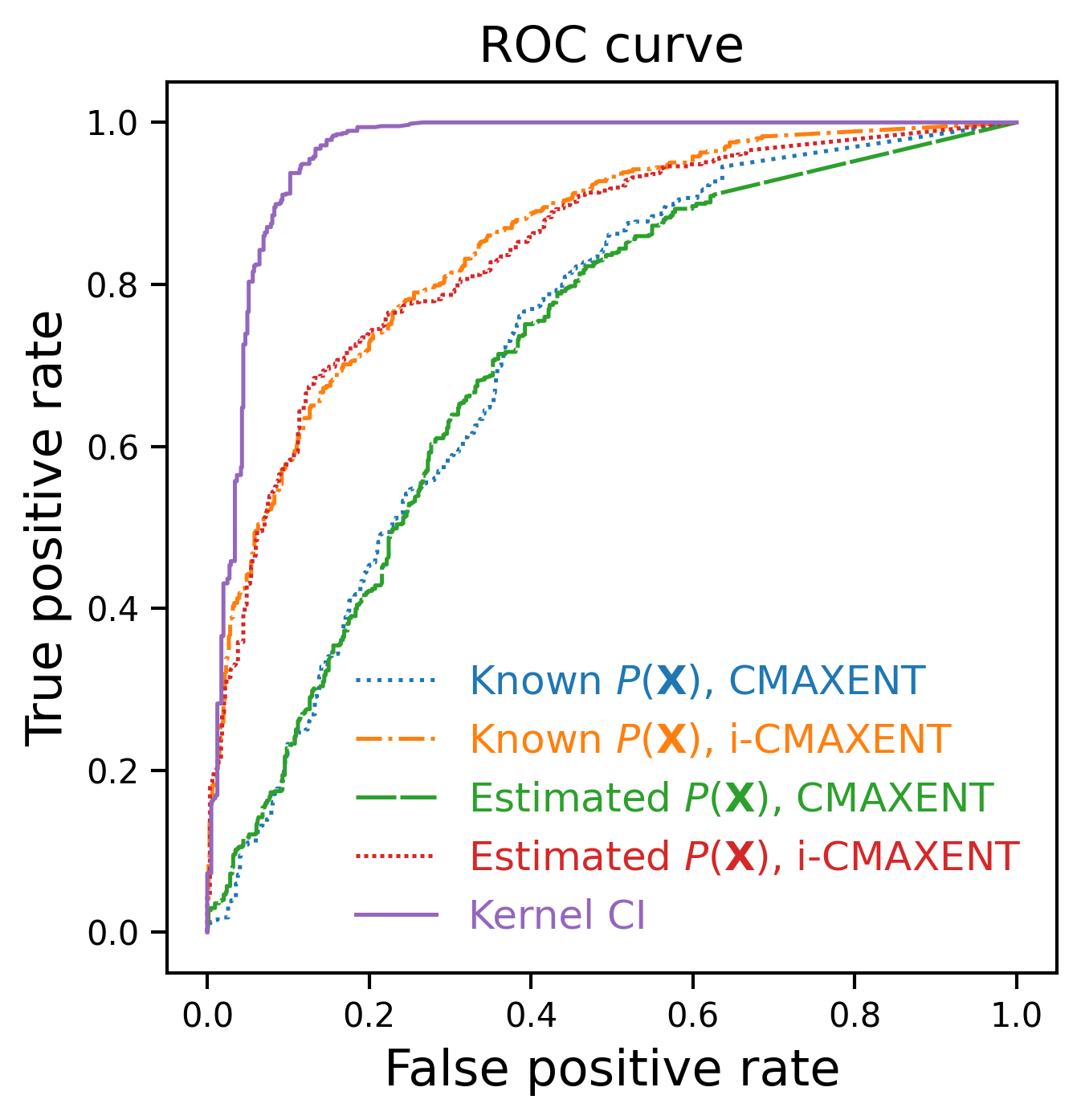}}
    
    \subfigure[ROC curve for varying number of interventional means for graph (a) \label{fig:combined_roc_a}]{\includegraphics[width=.33\textwidth]{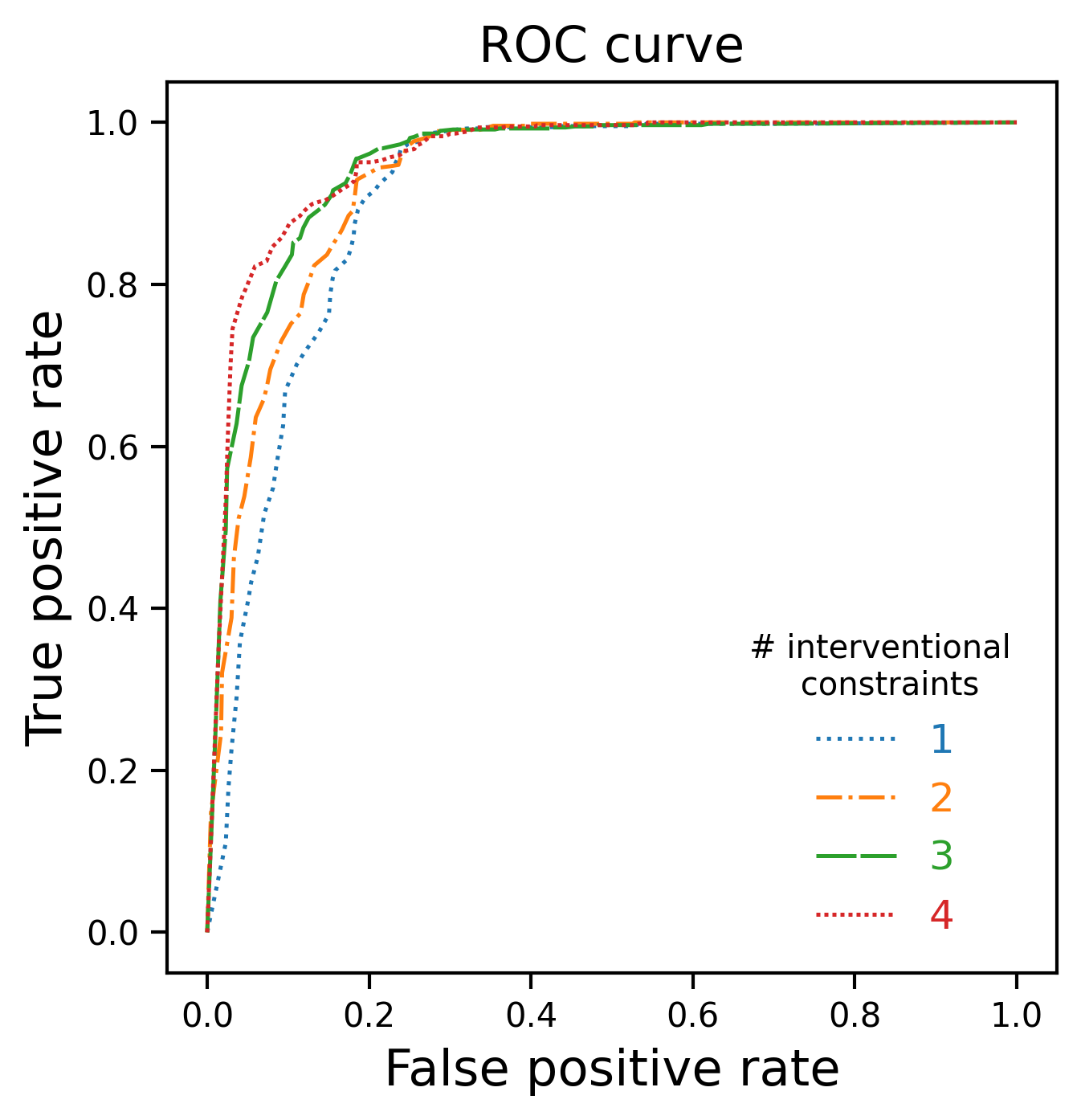}}
    \subfigure[ROC curve for varying number of interventional means for graph (b) \label{fig:combined_roc_b}]{\includegraphics[width=.33\textwidth]{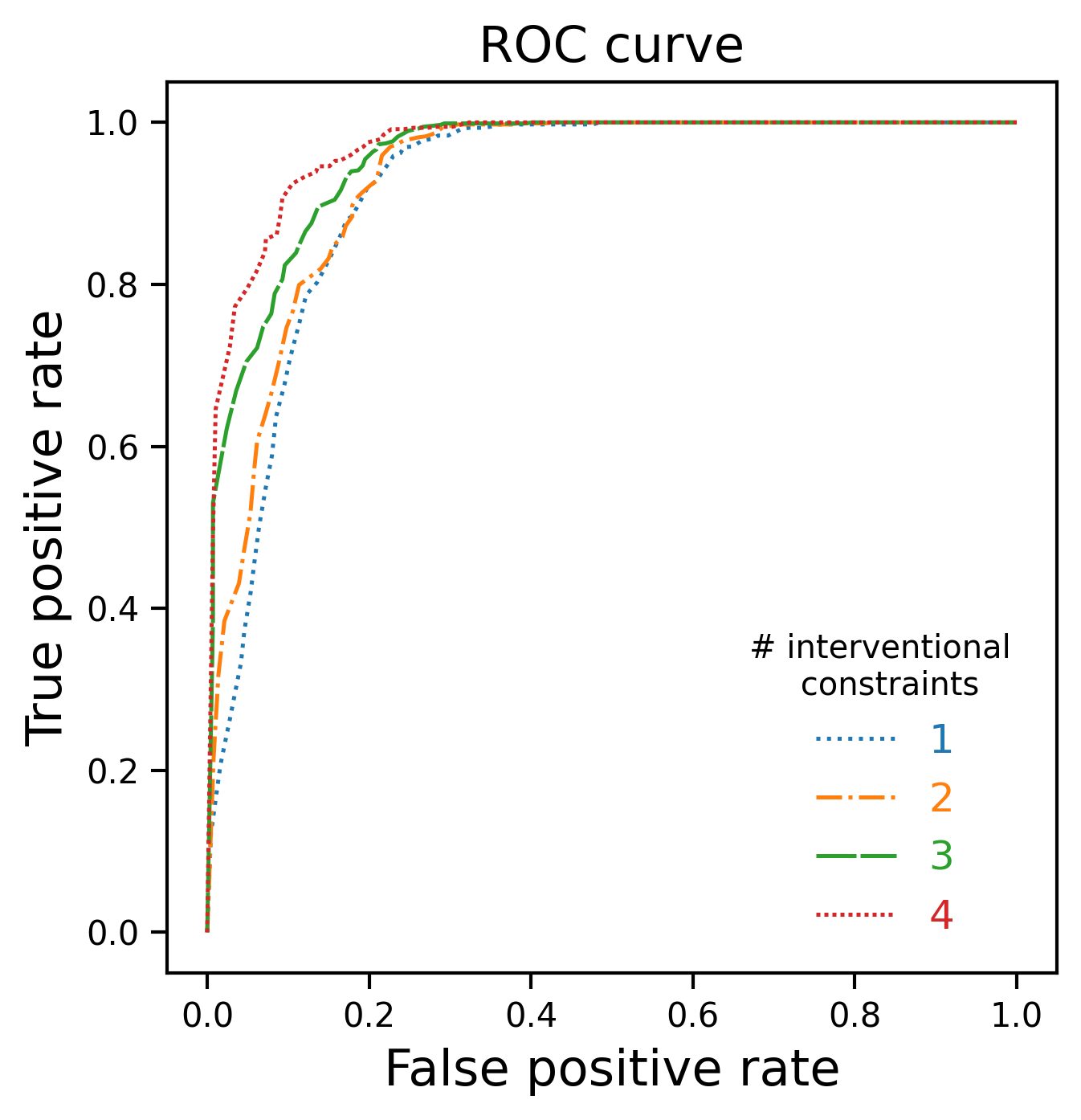}}
    \subfigure[ROC curve for varying number of interventional means for graph (c) \label{fig:combined_roc_c}]{\includegraphics[width=.33\textwidth]{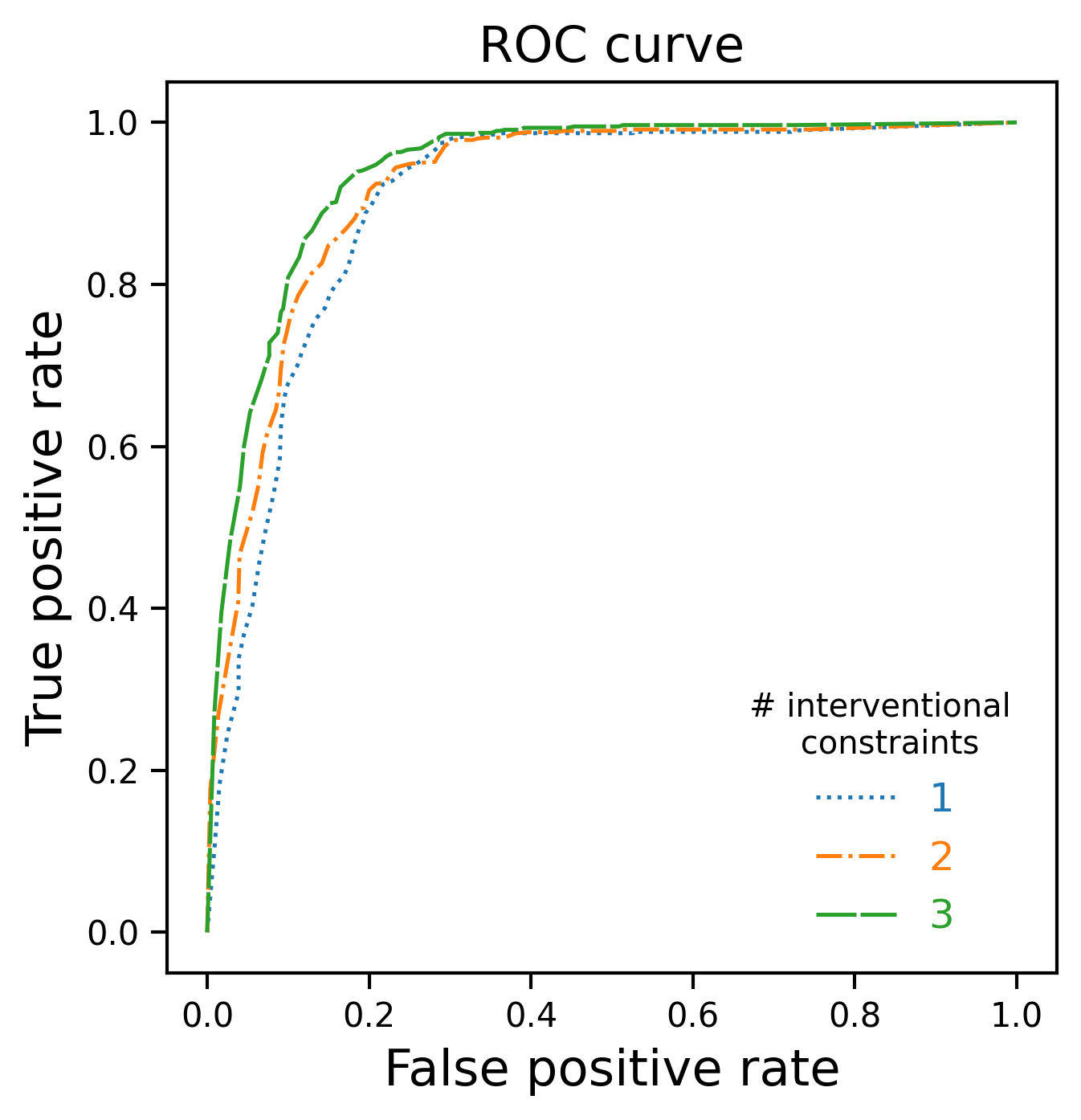}}
    
    \caption{
    Results for parental discovery.
    (\cref{fig:roc_overlay_a}), (\cref{fig:roc_overlay_b}), and (\cref{fig:roc_overlay_c}) show ROC curves for the identification of causal edges between the $X_i$s and $Y$ in setting 1. 
    For {\icmaxent,} we use constraints on all five single-interventional densities $\pdf(Y \mid \ddo(X_i))$. 
    For CMAXENT, we use constraints on the five single-conditional densities $\pdf(Y\mid X_i)$. 
    The KCI test has access to observations from the joint $\pdf(\bX, Y)$.
    For \icmaxent and CMAXENT, we consider two cases:
    (i) The joint observational density of the causes $\pdf(\bX)$ is known (blue and orange line).
    (ii) $\pdf(\bX)$ is estimated (green and red lines) from constraints on the five marginal densities $\pdf(X_i)$.
    Although our approach only uses single-variable interventional constraints as input, it achieves similar performance to the KCI test that uses the full generated dataset.}
\end{figure*}

\textbf{Parental discovery}
\Cref{fig:roc_overlay_a,fig:roc_overlay_b,fig:roc_overlay_c} show the ROC curves corresponding to the graphs in \Cref{fig:graph_exp_a,fig:graph_exp_b,fig:graph_exp_c}, respectively. 
Across all cases, \icmaxent consistently outperforms CMAXENT, and in graphs \Cref{fig:graph_exp_a} and \Cref{fig:graph_exp_b} achieves accuracy close to KCI despite using far less information.
Even for graph \Cref{fig:graph_exp_c} where the potential causes interact with each other, we observe that the difference between \icmaxent and KCI remains modest.

\Cref{fig:combined_roc_a,fig:combined_roc_b,fig:combined_roc_c} show the performance of \icmaxent in scenarios where information about the interventional distributions is provided as constraints only for a fraction of the potential causes, while for the remaining only conditional distributions are used.
We observe that the causes are recovered better as we increase the share of variables for which interventional data is provided.

\begin{figure*}[t!]
    \centering
    \includegraphics[width=\textwidth]{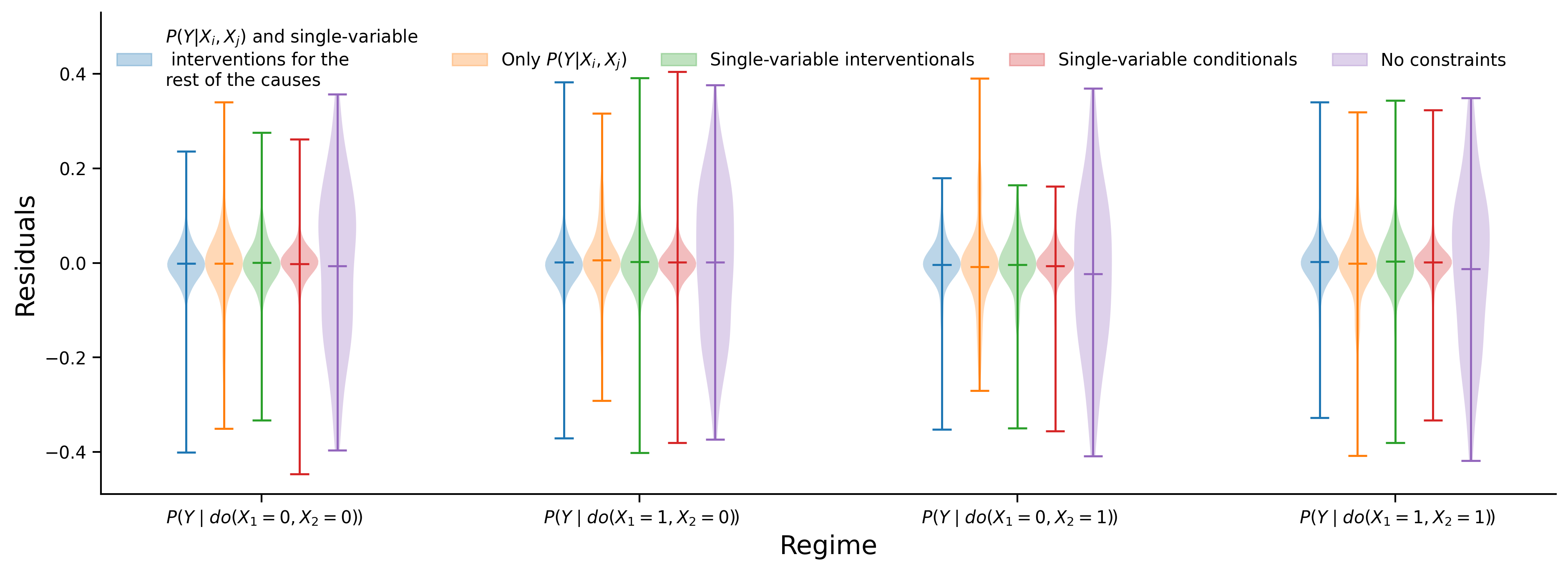}
    \caption{
    Residuals between true and estimated joint interventional distributions.
    The violin plots show the residuals
    between the true and the estimated joint interventional densities $\pdf(Y \mid \ddo(X_1, X_2))$ for five cases that differ in the constraints we use in the estimation. 
    The constraints are: 
    (i) the joint conditional $\pdf(Y\mid X_i, X_j)$ for a randomly chosen pair $X_i,X_j$, and single-variable interventionals $\pdf(Y\mid \ddo(X_k))$ for the rest of the variables (blue);
    (ii) only $\pdf(Y\mid X_i, X_j)$ (orange);
    (iii) single-variable interventionals $\pdf(Y\mid \ddo(X_k))$ for all causes (green);
    (iv) single-variable conditionals $\pdf(Y\mid X_k)$ for all causes (red); and
    (v) no constraints at all (purple).}
    \label{fig:violin_true_vs_estimated}
\end{figure*}

\textbf{Joint interventional distributions from single variable interventions}
\Cref{fig:violin_true_vs_estimated} depicts the  residuals between the estimated joint interventional distributions and the true joint interventional distributions, for interventions on $X_1$ and $X_2$.
Any added constraint improves estimation, with \icmaxent reducing residuals to around $5\%$ even when relying only on either single-variable conditional or interventional information.

Comparing residuals across settings, using only $\pdf(Y\mid X_i, X_j)$ yields the highest variance, while single-variable conditionals give slightly lower variance than the other cases.
Moreover, the extrema (the horizontal marks at the end of each distribution) of the residuals have similar spreads, depending on the regime.

\section{Discussion}
\label{sec:discussion}

\textbf{Causal marginal problem with interventional distributions}
We proved in \ref{th:interventionalMaxent} that the causal marginal problem can be solved using interventional distributions as constraints within the Maximum Entropy framework, yielding solutions in the exponential family. 
This extends CMAXENT to \icmaxent and shows that interventional information from experimental data on subsets of variables can be combined without requiring joint observations. 
\icmaxent therefore provides a principled tool for applications where the joint effect of disjoint treatments must be estimated from limited interventional data.

This extension enables two key tasks: 
parental discovery and estimation of joint interventional distributions. 
In feature selection, \icmaxent achieves performance close to methods with full joint data (KCI, \citep{zhang2011kernel}), while operating under the harder marginal setting. 
For joint interventional distribution estimation, \icmaxent with single-variable interventions performs comparably to using only single-variable conditionals. 
Its advantage is that it can directly incorporate experimental data when available, without distributional constraints \citep{saengkyongam2020learning} or functional form constraints \citep{kekic2025learning}.

\textbf{Limitations}
Our method requires knowing which variables were intervened upon to match expectations in the optimisation problem. 
This is inherent to the Maximum Entropy framework: 
constraints must specify the intervention. 
By contrast, methods like ICP only assume interventions occurred on some subset, but they require access to full joint observational data, which we do not. 
Furthermore, identifiability remains an important question.
It is clear that not all joint conditional distributions can be identified using the method, regardless of the empirical success we observe in \Cref{sec:results}.
We conjecture that some potential functional form constraints, such as the Generalised Additive Model (GAM) \citep{hastie1986generalized} assumption in \citep{kekic2025learning} might be enough for identification with \icmaxent.

\acks{This work was partly supported by the Amazon Hub project PSMEWE A03B “Merging data sources”.}

\bibliography{CausalityInterventionalCMAXENT.bib}

@article{deming1940least,
  title={On a least squares adjustment of a sampled frequency table when the expected marginal totals are known},
  author={Deming, W Edwards and Stephan, Frederick F},
  journal={The Annals of Mathematical Statistics},
  volume={11},
  number={4},
  pages={427--444},
  year={1940},
  publisher={JSTOR}
}

@article{jaynes1957information,
  title={Information theory and statistical mechanics},
  author={Jaynes, Edwin T},
  journal={Physical review},
  volume={106},
  number={4},
  pages={620},
  year={1957},
  publisher={APS}
}

@article{kellerer1964,
  doi = {10.1007/bf01360315},
  url = {https://doi.org/10.1007/bf01360315},
  year = {1964},
  month = jun,
  publisher = {Springer Science and Business Media {LLC}},
  volume = {153},
  number = {3},
  pages = {168--198},
  author = {Hans G. Kellerer},
  title = {Ma{\ss}theoretische Marginalprobleme},
  journal = {Mathematische Annalen}
}

@article{hastie1986generalized,
  title={Generalized additive models},
  author={Hastie, Trevor and Tibshirani, Robert},
  journal={Statistical science},
  volume={1},
  number={3},
  pages={297--310},
  year={1986},
  publisher={Institute of Mathematical Statistics}
}

@article{berger1996maximum,
  title={A maximum entropy approach to natural language processing},
  author={Berger, Adam and Della Pietra, Stephen A and Della Pietra, Vincent J},
  journal={Computational linguistics},
  volume={22},
  number={1},
  pages={39--71},
  year={1996}
}

@inproceedings{cooper1999causal,
  title={Causal discovery from a mixture of experimental and observational data},
  author={Cooper, Gregory F and Yoo, Changwon},
  booktitle={Proceedings of the Fifteenth conference on Uncertainty in artificial intelligence},
  pages={116--125},
  year={1999}
}

@inproceedings{tian2001causal,
  title={Causal discovery from changes},
  author={Tian, Jin and Pearl, Judea},
  booktitle={Proceedings of the Seventeenth conference on Uncertainty in artificial intelligence},
  pages={512--521},
  year={2001}
}

@book{tian2002general,
  title={A general identification condition for causal effects},
  author={Tian, Jin and Pearl, Judea},
  year={2002},
  publisher={eScholarship, University of California}
}

@book{jaynes2003probability,
  title={Probability theory: The logic of science},
  author={Jaynes, Edwin T},
  year={2003},
  publisher={Cambridge university press}
}

@article{grunwald2004game,
  title={Game theory, maximum entropy, minimum discrepancy and robust Bayesian decision theory},
  author={Gr{\"u}nwald, Peter D and Dawid, A Philip},
  journal={the Annals of Statistics},
  volume={32},
  number={4},
  pages={1367--1433},
  year={2004},
  publisher={Institute of Mathematical Statistics}
}

@inproceedings{sun2006causal,
  title={Causal inference by choosing graphs with most plausible Markov kernels},
  author={Sun, Xiaohai and Janzing, Dominik and Sch{\"o}lkopf, Bernhard},
  booktitle={Ninth International Symposium on Artificial Intelligence and Mathematics (AIMath 2006)},
  pages={1--11},
  year={2006}
}

@InProceedings{eaton2007exact,
  title = 	 {Exact Bayesian structure learning from uncertain interventions},
  author = 	 {Eaton, Daniel and Murphy, Kevin},
  booktitle = 	 {Proceedings of the Eleventh International Conference on Artificial Intelligence and Statistics},
  pages = 	 {107--114},
  year = 	 {2007},
  editor = 	 {Meila, Marina and Shen, Xiaotong},
  volume = 	 {2},
  series = 	 {Proceedings of Machine Learning Research},
  address = 	 {San Juan, Puerto Rico},
  month = 	 {21--24 Mar},
  publisher =    {PMLR},
  url = 	 {https://proceedings.mlr.press/v2/eaton07a.html}
}

@article{danks2008integrating,
  title={Integrating locally learned causal structures with overlapping variables},
  author={Danks, David and Glymour, Clark and Tillman, Robert},
  journal={Advances in Neural Information Processing Systems},
  volume={21},
  year={2008}
}

@article{wainwright2008graphical,
  title={Graphical models, exponential families, and variational inference},
  author={Wainwright, Martin J and Jordan, Michael I and others},
  journal={Foundations and Trends{\textregistered} in Machine Learning},
  volume={1},
  number={1--2},
  pages={1--305},
  year={2008},
  publisher={Now Publishers, Inc.}
}

@book{koller2009probabilistic,
  title={Probabilistic graphical models: principles and techniques},
  author={Koller, Daphne and Friedman, Nir},
  year={2009},
  publisher={MIT press}
}

@article{janzing2009distinguishing,
  title={Distinguishing cause and effect via second order exponential models},
  author={Janzing, Dominik and Sun, Xiaohai and Sch{\"o}lkopf, Bernhard},
  journal={arXiv preprint arXiv:0910.5561},
  year={2009}
}

@inproceedings{tillman2009structure,
  title={Structure learning with independent non-identically distributed data},
  author={Tillman, Robert E},
  booktitle={Proceedings of the 26th Annual International Conference on Machine Learning},
  pages={1041--1048},
  year={2009}
}

@book{pearl2009causality,
  title={Causality},
  author={Pearl, Judea},
  year={2009},
  publisher={Cambridge university press}
}

@inproceedings{zhang2011kernel,
  title={Kernel-based Conditional Independence Test and Application in Causal Discovery},
  author={Zhang, K and Peters, J and Janzing, D and Sch{\"o}lkopf, B},
  booktitle={27th Conference on Uncertainty in Artificial Intelligence (UAI 2011)},
  pages={804--813},
  year={2011},
  organization={AUAI Press}
}

@inproceedings{tillman2011learning,
  title={Learning equivalence classes of acyclic models with latent and selection variables from multiple datasets with overlapping variables},
  author={Tillman, Robert and Spirtes, Peter},
  booktitle={Proceedings of the Fourteenth International Conference on Artificial Intelligence and Statistics},
  pages={3--15},
  year={2011},
  organization={JMLR Workshop and Conference Proceedings}
}

@inproceedings{scholkopf2012causal,
  title={On causal and anticausal learning},
  author={Sch{\"o}lkopf, Bernhard and Janzing, Dominik and Peters, Jonas and Sgouritsa, Eleni and Zhang, Kun and Mooij, Joris},
  booktitle={Proceedings of the 29th International Coference on International Conference on Machine Learning},
  pages={459--466},
  year={2012}
}

@article{triantafillou2015constraint,
  title={Constraint-based causal discovery from multiple interventions over overlapping variable sets},
  author={Triantafillou, Sofia and Tsamardinos, Ioannis},
  journal={The Journal of Machine Learning Research},
  volume={16},
  number={1},
  pages={2147--2205},
  year={2015},
  publisher={JMLR. org}
}

@article{peters2016causal,
  title={Causal inference by using invariant prediction: identification and confidence intervals},
  author={Peters, Jonas and B{\"u}hlmann, Peter and Meinshausen, Nicolai},
  journal={Journal of the Royal Statistical Society: Series B (Statistical Methodology)},
  volume={78},
  number={5},
  pages={947--1012},
  year={2016},
  publisher={Wiley Online Library}
}

@article{farnia2016minimax,
  title={A minimax approach to supervised learning},
  author={Farnia, Farzan and Tse, David},
  journal={Advances in Neural Information Processing Systems},
  volume={29},
  year={2016}
}

@book{pearl2018theBook,
author = {Pearl, Judea and Mackenzie, Dana},
title = {The Book of Why: The New Science of Cause and Effect},
year = {2018},
isbn = {046509760X},
publisher = {Basic Books, Inc.},
address = {USA},
edition = {1st},
}

@article{heinze2018invariant,
  title={Invariant causal prediction for nonlinear models},
  author={Heinze-Deml, Christina and Peters, Jonas and Meinshausen, Nicolai},
  journal={Journal of Causal Inference},
  volume={6},
  number={2},
  year={2018},
  publisher={De Gruyter}
}

@software{jax2018github,
  author = {James Bradbury and Roy Frostig and Peter Hawkins and Matthew James Johnson and Chris Leary and Dougal Maclaurin and George Necula and Adam Paszke and Jake Vander{P}las and Skye Wanderman-{M}ilne and Qiao Zhang},
  title = {{JAX}: composable transformations of {P}ython+{N}um{P}y programs},
  url = {http://github.com/google/jax},
  version = {0.3.13},
  year = {2018},
}

@inproceedings{dhir2020integrating,
  title={Integrating overlapping datasets using bivariate causal discovery},
  author={Dhir, Anish and Lee, Ciar{\'a}n M},
  booktitle={Proceedings of the AAAI Conference on Artificial Intelligence},
  volume={34},
  number={04},
  pages={3781--3790},
  year={2020}
}

@inproceedings{saengkyongam2020learning,
  title={Learning joint nonlinear effects from single-variable interventions in the presence of hidden confounders},
  author={Saengkyongam, Sorawit and Silva, Ricardo},
  booktitle={Conference on Uncertainty in Artificial Intelligence},
  pages={300--309},
  year={2020},
  organization={PMLR}
}

@article{mooij2020joint,
  title={Joint causal inference from multiple contexts},
  author={Mooij, Joris M and Magliacane, Sara and Claassen, Tom},
  journal={The Journal of Machine Learning Research},
  volume={21},
  number={1},
  pages={3919--4026},
  year={2020},
  publisher={JMLRORG}
}

@Article{tikka2021causal,
    title = {Causal Effect Identification from Multiple Incomplete Data
    Sources: A General Search-Based Approach},
    author = {Santtu Tikka and Antti Hyttinen and Juha Karvanen},
    journal = {Journal of Statistical Software},
    year = {2021},
    volume = {99},
    number = {5},
    pages = {1--40},
    doi = {10.18637/jss.v099.i05},
}

@article{janzing2021causal,
  title={Causal versions of maximum entropy and principle of insufficient reason},
  author={Janzing, Dominik},
  journal={Journal of Causal Inference},
  volume={9},
  number={1},
  pages={285--301},
  year={2021},
  publisher={De Gruyter}
}

@inproceedings{garrido2022obtaining,
  title={Obtaining causal information by merging datasets with maxent},
  author={{Garrido Mejia}, {Sergio Hernan} and Kirschbaum, Elke and Janzing, Dominik},
  booktitle={International Conference on Artificial Intelligence and Statistics},
  pages={581--603},
  year={2022},
  organization={PMLR}
}

@inproceedings{gresele2022causal,
  title={Causal inference through the structural causal marginal problem},
  author={Gresele, Luigi and Von K{\"u}gelgen, Julius and K{\"u}bler, Jonas and Kirschbaum, Elke and Sch{\"o}lkopf, Bernhard and Janzing, Dominik},
  booktitle={International Conference on Machine Learning},
  pages={7793--7824},
  year={2022},
  organization={PMLR}
}

@article{sani2023bounding,
  title={Bounding probabilities of causation through the causal marginal problem},
  author={Sani, Numair and Mastakouri, Atalanti A and Janzing, Dominik},
  journal={arXiv preprint arXiv:2304.02023},
  year={2023}
}

@article{jeunen2022disentangling,
  title={Disentangling causal effects from sets of interventions in the presence of unobserved confounders},
  author={Jeunen, Olivier and Gilligan-Lee, Ciar{\'a}n and Mehrotra, Rishabh and Lalmas, Mounia},
  journal={Advances in Neural Information Processing Systems},
  volume={35},
  pages={27850--27861},
  year={2022}
}

@article{elahi2024identification,
  title={Identification of Average Causal Effects in Confounded Additive Noise Models},
  author={Elahi, Muhammad Qasim and Ghasemi, Mahsa and Kocaoglu, Murat},
  journal={arXiv preprint arXiv:2407.10014},
  year={2024}
}

@article{gimenez2022causal,
  title={Causal aggregation: estimation and inference of causal effects by constraint-based data fusion},
  author={Gimenez, Jaime Roquero and Rothenh{\"a}usler, Dominik},
  journal={Journal of Machine Learning Research},
  volume={23},
  number={335},
  pages={1--60},
  year={2022}
}

@article{kekic2025learning,
  title={Learning Joint Interventional Effects from Single-Variable Interventions in Additive Models},
  author={Keki{\'c}, Armin and Mejia, Sergio Hernan Garrido and Sch{\"o}lkopf, Bernhard},
  journal={arXiv preprint arXiv:2506.04945},
  year={2025}
}

@article{shi2023data,
  title={Data integration in causal inference},
  author={Shi, Xu and Pan, Ziyang and Miao, Wang},
  journal={Wiley Interdisciplinary Reviews: Computational Statistics},
  volume={15},
  number={1},
  pages={e1581},
  year={2023},
  publisher={Wiley Online Library}
}

@article{colnet2024causal,
  title={Causal inference methods for combining randomized trials and observational studies: a review},
  author={Colnet, B{\'e}n{\'e}dicte and Mayer, Imke and Chen, Guanhua and Dieng, Awa and Li, Ruohong and Varoquaux, Ga{\"e}l and Vert, Jean-Philippe and Josse, Julie and Yang, Shu},
  journal={Statistical science},
  volume={39},
  number={1},
  pages={165--191},
  year={2024},
  publisher={Institute of Mathematical Statistics}
}

@inproceedings{bareinboim2012causal,
  title={Causal inference by surrogate experiments: z-identifiability},
  author={Bareinboim, Elias and Pearl, Judea},
  booktitle={Proceedings of the Twenty-Eighth Conference on Uncertainty in Artificial Intelligence},
  pages={113--120},
  year={2012}
}

@inproceedings{lee2020general,
  title={General identifiability with arbitrary surrogate experiments},
  author={Lee, Sanghack and Correa, Juan D and Bareinboim, Elias},
  booktitle={Uncertainty in artificial intelligence},
  pages={389--398},
  year={2020},
  organization={PMLR}
}

@article{jung2023estimating,
  title={Estimating causal effects identifiable from a combination of observations and experiments},
  author={Jung, Yonghan and D{\'\i}az, Iv{\'a}n and Tian, Jin and Bareinboim, Elias},
  journal={Advances in Neural Information Processing Systems},
  volume={36},
  pages={46446--46490},
  year={2023}
}

@article{hindersah2022rice,
  title={Rice yield grown in different fertilizer combination and planting methods: Case study in Buru Island, Indonesia},
  author={Hindersah, Reginawanti and Kalay, Agusthinus Marthin and Talahaturuson, Abraham},
  journal={Open Agriculture},
  volume={7},
  number={1},
  pages={871--881},
  year={2022},
  publisher={De Gruyter Open Access}
}

@article{qiu2022effect,
  title={Effect of Irrigation and Fertilizer Management on Rice Yield and Nitrogen Loss: A Meta-Analysis},
  author={Qiu, Haonan and Yang, Shihong and Jiang, Zewei and Xu, Yi and Jiao, Xiyun},
  journal={Plants},
  volume={11},
  number={13},
  pages={1690},
  year={2022},
  publisher={MDPI}
}

@article{wihardjaka2022effect,
  title={Effect of fertilizer management on potassium dynamics and yield of rainfed lowland rice in Indonesia},
  author={Wihardjaka, Anicetus and Harsanti, Elisabeth Srihayu and Ardiwinata, Asep Nugraha},
  journal={Chilean journal of agricultural research},
  volume={82},
  number={1},
  pages={33--43},
  year={2022},
  publisher={SciELO Chile}
}

@article{uhler2013geometry,
  title={Geometry of the faithfulness assumption in causal inference},
  author={Uhler, Caroline and Raskutti, Garvesh and B{\"u}hlmann, Peter and Yu, Bin},
  journal={The Annals of Statistics},
  pages={436--463},
  year={2013},
  publisher={JSTOR}
}

@inproceedings{zhang2002strong,
  title={Strong faithfulness and uniform consistency in causal inference},
  author={Zhang, Jiji and Spirtes, Peter},
  booktitle={Proceedings of the Nineteenth conference on Uncertainty in Artificial Intelligence},
  pages={632--639},
  year={2002}
}

@article{zhang2008detection,
  title={Detection of unfaithfulness and robust causal inference},
  author={Zhang, Jiji and Spirtes, Peter},
  journal={Minds and Machines},
  volume={18},
  pages={239--271},
  year={2008},
  publisher={Springer}
}

@book{spirtes2000causation,
  title={Causation, prediction, and search},
  author={Spirtes, Peter and Glymour, Clark N and Scheines, Richard},
  year={2000},
  publisher={MIT press}
}

@inproceedings{olko2025since,
  title={Since Faithfulness Fails: The Performance Limits of Neural Causal Discovery},
  author={Olko, Mateusz and Gajewski, Mateusz and Wojciechowska, Joanna and Morzy, Miko{\l}aj and Sankowski, Piotr and Mi{\l}o{\'s}, Piotr},
  booktitle={International Conference on Machine Learning},
  pages={47155--47175},
  year={2025},
  organization={PMLR}
}

\appendix

\section{Proofs of the main results}\label{app:proofs}

\nonparametricIdentification*

\begin{proof}
    We have
    \begin{align}
        \begin{split}\label{eq:jointConditional1}
            \pdf(Y\mid X_1=0) 
            =& \pdf(Y\mid X_1=0, X_2=0) \pdf(X_2=0\mid X_1=0)\\
            &+ \pdf(Y\mid X_1=0, X_2=1) \pdf(X_2=1\mid X_1=0)
        \end{split}\\
        \begin{split}\label{eq:jointConditional2}
            \pdf(Y\mid X_1=1) 
            =& \pdf(Y\mid X_1=1, X_2=0) \pdf(X_2=0\mid X_1=1)\\
            &+ \pdf(Y\mid X_1=1, X_2=1) \pdf(X_2=1\mid X_1=1)
        \end{split}\\
        \begin{split}\label{eq:jointConditional3}
            \pdf(Y\mid X_2=0) 
            =& \pdf(Y\mid X_1=0, X_2=0) \pdf(X_1=0\mid X_2=0)\\
            &+ \pdf(Y\mid X_1=1, X_2=0) \pdf(X_1=1\mid X_2=0)
        \end{split}\\
        \begin{split}\label{eq:jointConditional4}
            \pdf(Y\mid X_2=1) 
            =& \pdf(Y\mid X_1=0, X_2=1) \pdf(X_1=0\mid X_2=1)\\
            &+ \pdf(Y\mid X_1=1, X_2=1) \pdf(X_1=1\mid X_2=1).
        \end{split}
    \end{align}
Using \Cref{eq:jointConditional1,eq:jointConditional2,eq:jointConditional3,eq:jointConditional4} we can find the following expressions:
    \begin{align}
        \pdf(Y\mid X_1=0, X_2=0)
        =& \frac{\pdf(Y\mid X_1=0) - \pdf(Y\mid X_1=0, X_2=1) \pdf(X_2=1\mid X_1=0)}{\pdf(X_2=0\mid X_1=0)}\\
        \pdf(Y\mid X_1=1, X_2=1) 
        =& \frac{\pdf(Y\mid X_1=1) - \pdf(Y\mid X_1=1, X_2=0) \pdf(X_2=0\mid X_1=1)}{\pdf(X_2=1\mid X_1=1)}\\
        \pdf(Y\mid X_1=1, X_2=0)
        =& \frac{\pdf(Y\mid X_2=0) - \pdf(Y\mid X_1=0, X_2=0) \pdf(X_1=0\mid X_2=0)}{\pdf(X_1=1\mid X_2=0)}\\
        \pdf(Y\mid X_1=0, X_2=1) 
        =& \frac{\pdf(Y\mid X_2=1) - \pdf(Y\mid X_1=1, X_2=1) \pdf(X_1=1\mid X_2=1)}{\pdf(X_1=0\mid X_2=1)},
    \end{align}
with which we can replace recursively to find the following expression for $\pdf(Y\mid X_1=0, X_2=0)$:
    \begin{align}
        \begin{split}
        \pdf(Y\mid X_1=0, X_2=0)
        =    \pdf &(Y\mid X_1=1, X_2=1) \pdf(X_2=0\mid X_1=0)^{-1}\\
                &[\pdf(Y\mid X_1=0) - 
                    \pdf(X_2=1\mid X_1=0)\pdf(X_1=0\mid X_2=1)^{-1}\\
                &\quad  [\pdf(Y\mid X_2=1) - 
                            \pdf(X_1=1\mid X_2=1)\pdf(X_2=1\mid X_1=1)^{-1}\\
                &\quad\quad     [\pdf(Y\mid X_1=1) -
                                    \pdf(X_1=0\mid X_2=0)\pdf(X_1=1\mid X_2=0)^{-1}\\
                &\quad\quad\quad        [\pdf(Y\mid X_2=0) - 
                                            \pdf(Y\mid X_1=0, X_2=0)]]]],
        \end{split}
    \end{align}
and after some algebra we can find an expression of $\pdf(Y\mid X_1=0, X_2=0)$ as a function of only bivariate distributions, which we can replace again in \Cref{eq:jointConditional1,eq:jointConditional2,eq:jointConditional3,eq:jointConditional4} to find the other distributions of interest.
\end{proof}

\interventionalMaxent*
\begin{proof}
We start by setting up the Lagrangian, where we have one $\lambda$ for each constraint in our optimisation problem.
\begin{align}
\begin{split}
    \mathcal{L} =& -\sum_{y,\bx}\pdf(y\mid\bx)\pdf(\bx)\log \pdf(y\mid\bx)\\
    &+ \sum_{k=1}^{K}\lambda_{k}(\sum_{y,\bx}\pdf(y\mid\bx)\pdf(\bx)f_{k}(y,\bx) - \tilde{f}_{k})\\
    &+ \sum_{j=1}^{J}\sum_{\bx_{S_{j}}}\lambda_{j}^{\bx_{S_{j}}}(\sum_{y}\pdf(y\mid \bx_{S_{j}})g_{j}(y,\bx_{S_{j}})-\tilde{g}_{j}(\bx_{S_{j}}))\\
    &+ \sum_{l=1}^{L}\sum_{\bx_{S^C_{l}}, \bx_{S^I_{l}}}\lambda_{l}^{\bx_{S^C_{l}}, \bx_{S^I_{l}}}(\sum_{y}\pdf(y\mid \ddo(\bx_{S^I_{l}}), \bx_{S^C_{l}})h_{l}(y,\bx_{S^C_{l}}, \bx_{S^I_{l}})-\tilde{h}_{l}(\bx_{S^C_{l}}, \bx_{S^I_{l}}))\\
    &+ \sum_{\bx}\lambda^{\bx}(\sum_{y}\pdf(y\mid\bx)-1)
\end{split}
\end{align}

The above Lagrangian contains an interventional distribution that depends on $\pdf(y \mid \bx)$.
In order to optimise the Lagrangian we need to replace $\pdf(y \mid \ddo(x))$ as a function of $\pdf(y \mid x)$.
In the following, we will use $s_{j}$ and $s^I_{l}$ to denote elements of the set $S_{j}$ and $S^I_{l}$.
In the following derivation, we will denote the complement of $S_j$ in $\bX$ as $S_j'$.
\begin{align}
    \pdf(y\mid \bx_{S_{j}})
    =& \sum_{\bx_{S_j'}} \pdf(y, \bx_{S_j'}\mid \bx_{S_{j}})\\
    =& \sum_{\bx_{S_j'}} \pdf(y\mid \bx_{S_j'}, \bx_{S_{j}})\pdf(\bx_{S_j'}\mid \bx_{S_{j}})\\
    =& \sum_{\bx_{S_j'}}\pdf(y\mid \bx)\pdf(\bx_{S_j'}\mid \bx_{S_{j}}).
\end{align}

Using this technique, we can express $\pdf(y\mid \ddo(\bx_{S^I_{l}}), \bx_{S^C_{l}})$ as a function of $\pdf(y \mid \bx)$, that is,
$\pdf(y\mid \ddo(\bx_{S^I_{l}}), \bx_{S^C_l}) = \pdf(y\mid\bx)\overline{\pdf(y\mid \ddo(\bx_{S^I_{l}}), \bx_{S^C_{l}})}$, where $\overline{\pdf(y\mid \ddo(\bx_{S^I_{l}}), \bx_{S^C_{l}})}$ is a combination of sums and products of known observable quantities.

Differentiating with respect to each ${\pdf(Y=y\mid \bX=\bx)}$ and all the multipliers, we obtain
\allowdisplaybreaks
\begin{align}\label{eq:derivativesLagrangian}
    \begin{split}
    \frac{\partial\mathcal{L}}{\partial \pdf(y\mid\bx)} = &-\pdf(\bx)[\log \pdf(y\mid\bx) + 1]\\*
    &+ \sum_{k=1}^{K}\lambda_{k}\pdf(\bx)f_{k}(y,\bx)\\*
    &+ \sum_{j=1}^{J}\sum_{\bx_{S_{j}}}\lambda_{j}^{\bx_{S_{j}}}\left[\sum_{\bx_{S_j^C}}\pdf(\bx_{S_j^C}\mid \bx_{S_{j}})g_{j}(y,\bx_{S_{j}})\right]\\*
    &+ \sum_{l=1}^{L}\sum_{\bx_{S^C_{l}}, \bx_{S^I_{l}}}\lambda_{l}^{\bx_{S^C_{l}}, \bx_{S^I_{l}}}\overline{\pdf(y\mid \ddo(\bx_{S^I_{l}}), \bx_{S^C_{l}})}h_{l}(y,\bx_{S^C_{l}}, \bx_{S^I_{l}})\\*
    &+ \lambda^{\bx}
    \end{split}\\
    \frac{\partial\mathcal{L}}{\partial \lambda_{k}} = &\sum_{y,\bx}\pdf(y\mid\bx)\pdf(\bx)f_{k}(y, \bx) - \tilde{f}_{k},\, \forall\, k\\
    \frac{\partial\mathcal{L}}{\partial \lambda_{j}^{\bx_{s_{j}}}} = &\sum_{y}\pdf(y\mid \bx_{s_{j}})g_{j}(y,\bx_{s_{j}})-\tilde{g}_{j}(\bx_{s_{j}}),\, \forall\, s_{j}\in S_{j},\, \forall\, j\\
    \frac{\partial\mathcal{L}}{\partial \lambda_{l}^{\bx_{s^C_{l}}, \bx_{s^I_{l}}}} = &\sum_{y}\pdf(y\mid \ddo(\bx_{S^I_{l}}), \bx_{S^C_{l}})h_{l}(y,\bx_{S^C_{l}}, \bx_{S^I_{l}})-\tilde{h}_{l}(\bx_{S^C_{l}}, \bx_{S^I_{l}}),\, \forall\, s^C_{l}\in S^C_{l}, \forall\, s^I_{l}\in S^I_{l}, \, \forall\, l\\
    \frac{\partial\mathcal{L}}{\partial \lambda^{\bx}} = &\sum_{y}\pdf(y\mid\bx)-1,\, \forall\, \bx
\end{align}
Notice that the derivative with respect to the Lagrangian in \Cref{eq:derivativesLagrangian}
is 0 for those functions $f_{k}, g_{j}$ and $h_{l}$ for which $\bx\notin S_{j}, S_{k}, S^C_{l}, S^I_{l}$.

We then find the solution to $\pdf(y\mid\bx)$ when the above equations are equal to 0. From \Cref{eq:derivativesLagrangian} we get
\begin{align}
\begin{split}
    \pdf(y\mid\bx) = \exp&\left[\sum_{k=1}^{K}\lambda_{k}f_{k}(y,\bx) 
    + \frac{1}{\pdf(\bx)}\sum_{j=1}^{J}\sum_{\bx_{S_{j}}}\lambda_{j}^{\bx_{S_{j}}}g_{j}(y,\bx_{S_{j}})\sum_{\bx_{S_j^c}}\pdf(\bx_{S_j^c}\mid \bx_{S_{j}})\right.\\
    &\left.+ \frac{1}{\pdf(\bx)}\sum_{l=1}^{L}\sum_{\bx_{S^C_{l}}, \bx_{S^I_{l}}}\lambda_{l}^{\bx_{S^C_{l}}, \bx_{S^I_{l}}}\overline{\pdf(y\mid \ddo(\bx_{S^I_{l}}), \bx_{S^C_{l}})}h_{l}(y,\bx_{S^C_{l}}, \bx_{S^I_{l}})
    + \frac{1}{\pdf(\bx)}\lambda^{\bx}\vphantom{\sum_{k=1}^{K}}\right].
\end{split}
\end{align}
Notice that this equation is well-defined as long as $\pdf(\bx)>0$ for all $\bx$.
Because the elements inside the exponential depend on $\bx$, we can rename $\lambda$ with $\tilde{\lambda}$. 
In addition, we gather the constants into the normalizing constant, which depends on $\bx$, giving us
\begin{align}\label{eq:interventionalExponentialFamily}
\begin{split}
    \pdf(y\mid\bx) 
    = \exp& \left[\sum_{k=1}^{K}\lambda_{k}f_{k}(y,\bx) 
    + \sum_{j=1}^{J}\sum_{\bx_{S_{j}}} \tilde{\lambda}_{j}^{\bx_{S_{j}}}g_{j}(y,\bx_{S_{j}}) \right.\\
    &+ \left. \sum_{l=1}^{L}\sum_{\bx_{S^C_{l}}}\sum_{\bx_{S^I_{l}}}\tilde{\lambda}_{l}^{\bx_{S^C_{l}}, \bx_{S^I_{l}}}h_{l}(y,\bx_{S^C_{l}}, \bx_{S^I_{l}}) 
    + \beta(\bx)\right]
\end{split}
\end{align}
as required.
\end{proof}

\identifiabilityIncomingArrows*
\begin{proof}
Because of the assumed generative process, the causal sufficiency assumption,
and the fact that $Y$ is the only potential child of $X_{j}$,
there are no ``bidirectional'' arrows connected to any child of $X_{j}$ in the sense of \cite{tian2002general}. As a result, the conditions of Theorem 2 in \cite{tian2002general} apply and $\pdf(Y\mid \ddo(X_{j}))$ is identifiable using observational quantities.

We have just proved the identifiability of the interventional distribution.
Now we would like to prove that the set 
$\bX'=\bX\setminus X_{j}$ is a valid adjustment set.
This is true for the following reasons.
First, the assumption of the three levels of the generative process, 
which has as a consequence that there is no collider 
(or descendants of a collider)
between $Y$ and the elements in $\bX'$.
Second, Assumption (1) which states 
there is no unobserved confounder between $\bX$ and $Y$,
thus $\bX'$ blocks any potential backdoor path 
between $X_{j}$ and $Y$ through any confounder.
\end{proof}

\section{Computation of Maxent through norm minimisation}
\label{app:normMinimisation}

As shown in \Cref{th:interventionalMaxent}, the interventional maximum entropy solution is equivalent to maximum likelihood where the dual variables are the parameters of the exponential family.
We will now show how the maximum likelihood problem can be expressed as the minimisation between the difference between the given empirical averages and the expectations of the functions given the exponential family distribution.
The way we maximise \Cref{eq:interventionalExponentialFamily} for all our data, with respect to the Lagrange multipliers (the $\lambda_{i}$), is by making equal to 0 the derivative of the log-likelihood with respect to the parameters:
\begin{align}
    \frac{\partial\log \pdf(Y=y\mid \bX=\bx)}{\partial \lambda_{i}} &=  
    \frac{1}{N}\sum_{n=1}^{N}f_{k}(y^{n},\bx^{n}) - \frac{f_{k}(y,\bx)\sum_{k}\lambda_{k}f_{k}(y,\bx)}{\exp\sum_{y}\sum_{k}\lambda_{k}f_{k}(y,\bx)}=0.
\end{align}
Because we have empirical averages, the previous equation (for the whole data) becomes
\begin{align}\label{eq:empiricalEqualsExpectation}
  \tilde{f}_{k}(y,\bx) - \mathbb{E}_{\pdf_{\lambda}}[f_{k}(y,\bx)]=0,
\end{align}
which we can compute by using any method that minimises the difference between the observed empirical average and the entailed expectation using the exponential family distribution.
That is, we can compute the solution to the dual problem of the exponential family with
\begin{align}
  \lambda &= \argmin_{\lambda}\lVert\tilde{f}_{k}(y,\bx) - \mathbb{E}_{\pdf_{\lambda}}[f_{k}(y,\bx)]\rVert
\end{align}

In the synthetic experiments Section (6), there were situations where $\pdf(\bx)=0$.
We fixed this by adding a machine epsilon to all possible combinations and renormalizing.
\section{Relative difference estimator}
\label{app:relativeDifference}

The relative difference estimator introduced in \citep{garrido2022obtaining} is an estimator of how close two parameters are to each other so that an analyst can decide whether there exist conditional independence. 
The relative difference estimator can take values between 0 and 1.
However, there is no probabilistic analysis of the estimator, so one cannot consider the value of the relative difference estimator as a well-calibrated probability of the difference between multipliers (or difference between a multiplier and 0).
The estimator is defined as
\begin{align}
\theta_i = \frac{\left|\lambda_i^{1}-\lambda_i^{2}\right|}{\max\{|\lambda_i^{1}|, |\lambda_i^{2}|, \left|\lambda_i^{1}-\lambda_i^{2}\right|, 1\}} \;\in\left[0,1\right] \; ,
\end{align}
where $\lambda_i^1, \lambda_i^2$ are the two Lagrange multipliers for the constraints associated with $X_i$.
\section{Details of data generation processes}
\label{app:detailsDGP}

The data generation process consists of two steps. For each of the three graph structures we first sample the parameters of a generative process.
For this we sample for each variable $Z\in\{\bU,\bX,Y\}$ a value for $\pdf(Z{=}1\mid PA_{Z})$ by sampling from a uniform distribution between $0.1$ and $0.9$. 
We do this for each combination of values of the parents $PA_{Z}$ of $Z$.
For $\bX$ and $\bU$ he parents are fixed by the respective graph structure, but for $Y$ we randomise the parents with probability $0.5$ for including any particular $X_i$ as a parent of $Y$.
Next, we sample observations using this generative process. 
That is, we do ancestral sampling where we sample each variable $Z$ using a Bernoulli distribution with the before chosen probability $\pdf(Z=1\mid PA_Z)$.
To obtain the data from the interventional distributions, we simply set the variable to the particular intervened value and then proceed in the generative process.

\section{Optimisation and convergence}
\label{app:convergence}

To find the Lagrange multipliers, using the norm minimisation procedure explained in \Cref{app:normMinimisation}, we use the Broyden–Fletcher–Goldfarb–Shanno (BFGS) algorithm  implemented in the JAX Python library \citep{jax2018github}.
We consider an estimation as converged if the norm (the sum of the squared of the residuals between the empirical averages given as constraints and the expectations entailed by the exponential family distribution) are less than $0.01$, or if the optimiser terminates successfully.
In some cases, convergence in this sense is not achieved on the first application of the optimisation algorithm.
When this happens, we simply use the optimisation algorithm again using the previously found Lagrange multipliers as the initial values for the optimisation.
We repeat this process until convergence in the above sense is achieved.
We found that in most cases one extra optimisation is enough to achieve convergence, there were four cases where we had to run the optimiser more than two times: twice for three, once for four, and once for five.
\section{Other faithfulness assumptions}
\label{app:faithfulness}

The main goal of faithfulness assumptions is to make conclusions about the causal graph from statistical conclusions.
The most important use of this is to do causal discovery from finite sample data.
We begin with one of the simplest notions of faithfulness \citep{spirtes2000causation}:
\begin{assumption}[Faithfulness]
\label{asmp:faithfulness}
A distribution $\cdf$ is faithful to a DAG $\cG$ if no conditional independence relations other than the ones entailed by the Markov property are present.
\end{assumption}

Since \Cref{asmp:faithfulness} does not depend parametrically on how we measure conditional independence, \citet{zhang2002strong} modify this notion to include a $\lambda$ parameter:
\begin{assumption}[$\lambda$-strong Faithfulness]
\label{asmp:lambda-strong-faithfulness}
Given $\lambda\in(0,1)$, a multivariate Gaussian distribution $\cdf$ is said to be $\lambda$-strong-faithful to a DAG $\cG=(V,E)$ if for any $i,j\in V$ and any $S\subseteq V\setminus\set{i,j}$ such that $j$ is d-separated from $i$ given $S$ if and only if $\absoluteValue{\Corr(X_i, X_j \mid X_S)} \leq \lambda$.
\end{assumption}

As \citet{uhler2013geometry} prove, and in contrast with the usual Faithfulness assumption, the set of distributions for which \Cref{asmp:lambda-strong-faithfulness} hold does not have Lebesgue measure zero.
There is an important implication:
it has been proven that the usual faithfulness assumption holds ``almost surely'' as the set of distributions for which the assumption does not hold has measure zero.
However, if this were to change, that is, if this set of distributions did not have Lebesgue measure zero, then there are uncountably many such distributions, and the possibility of being in such scenario is non-zero.
Of course if faithfulness does not hold, the causal graphs built from conditional independence tests might be wrong.
More recently \citet{olko2025since} studied strong faithfulness in the context of neural causal discovery and non linear models, confirming the results found by \citet{uhler2013geometry} in the linear case and the PC algorithm.

Other notions of faithfulness have been introduced in \citet{zhang2008detection}:
\begin{assumption}[Restricted $\lambda$-strong Faithfulness]
\label{asmp:resctricted-lambda-strong-faithfulness}
Given $\lambda\in(0,1)$, a multivariate Gaussian distribution $\cdf$ is said to be restricted $\lambda$-strong-faithful to a DAG $\cG=(V,E)$ if both of the following hold:
\begin{description}
    \item[(i)] 
    $\min\set{
    \absoluteValue{\Corr(X_i,X_j\mid X_S)}: 
    (i,j)\in E,\, 
    S\subseteq V\setminus\set{i,j}, 
    \text{ such that } \absoluteValue{S}\leq 
    \text{deg}(\cG)} > \lambda$, where $\text{deg}(G)$ denotes the sum of the indegree and the outdegree of nodes in $\cG$. 
    This condition is also known as Adjacency faithfulness \citep{zhang2008detection}, according to \citet{uhler2013geometry}.
    \item[(ii)]
    $\min\set{\absoluteValue{
    \Corr(X_i,X_j\mid X_S)}: 
    (i,j,S)\in N_\cG} > \lambda$, where $N_\cG$ is the set of triples $(i,j,S)$ such that $i,j$ are not adjacent but there exists $k\in V$ making $(i,j,k)$ an unshielded triple (a triple of the form $i\rightarrow k \leftarrow j$) and $i,j$ are not d-separated given $S$. 
    This condition is also known as Orientation faithfulness \citep{zhang2008detection}, according to \citet{uhler2013geometry}.
\end{description}
\end{assumption}

\begin{assumption}[Triangle Faithfulness]
\label{asmp:triangle-faithfulness}
Suppose the true causal DAG of a set of nodes $V$ is $\cG$.
Let $X,Y,Z$ be any three variables that form a triangle in $\cG$ (i.e., each pair of vertices is adjacent):
\begin{description}
    \item[(i)]
    If $Y$ is a noncollider on the path $X\mathdash Y \mathdash Z$, then $X$ and $Z$ are not independent conditional on any subset of $V\setminus\set{X,Z}$ that does not contain $Y$.
    \item[(ii)]
    If $Y$ is a collider on the path, (i.e., $X\rightarrow Y \leftarrow Z$), then $X$ and $Z$ are not independent conditional on any subset of $V\setminus\set{X,Z}$ that contains $Y$.
\end{description}
\end{assumption}

Out of these definitions of faithfulness only \Cref{asmp:faithfulness,asmp:lambda-strong-faithfulness} have been studied to a certain degree, given the difficulty of analysing these assumptions. 
We believe that the assumption of faithful f-expectations possess the same genericity as \Cref{asmp:faithfulness}.

\end{document}